\documentclass[12pt,a4paper]{amsart}
\usepackage{amsmath}
\usepackage{amsthm}
\usepackage{amssymb}
\usepackage{hyperref}
\usepackage{bbm}
\usepackage{braket}

\usepackage{tikz}
\usetikzlibrary{decorations.pathreplacing}

\usepackage{vmargin}
\setpapersize{A4}                       

\newtheorem{Theo}{Theorem}[section]
\newtheorem{Prop}[Theo]{Proposition}
\newtheorem{Coro}[Theo]{Corollary}
\newtheorem{Lemm}[Theo]{Lemma}
\newtheorem{Defi}[Theo]{Definition}

\theoremstyle{definition}

\title[Locality estimes for complex time evolution in 1D]{Locality estimates for\\ complex time evolution in 1D}

\author[D. P\'{e}rez-Garc\'{i}a]{David P\'{e}rez-Garc\'{i}a}

\address[P\'{e}rez-Garc\'{i}a]{\newline
Instituto de Ciencias Matem\'{a}ticas (CSIC-UAM-UC3M-UCM), C/ Nicol\'{a}s Cabrera 13-15, Campus de Cantoblanco,
28040 Madrid, Spain\newline Departamento de An\'{a}lisis y Matem\'{a}tica Aplicada, Universidad Complutense de Madrid,
28040 Madrid, Spain}
\email{dperezga@ucm.es}
\urladdr{https://orcid.org/0000-0003-2990-791X}

\author[A. P\'{e}rez-Hern\'{a}ndez]{Antonio P\'{e}rez-Hern\'{a}ndez}

\address[P\'{e}rez-Hern\'{a}ndez]{\newline Departamento de An\'{a}lisis y Matem\'{a}tica Aplicada, Universidad Complutense de Madrid,
28040 Madrid, Spain}
 \email{antonp07@ucm.es}
\urladdr{https://orcid.org/0000-0001-8600-7083}

\thanks{
This project has been partially supported by the European Research Council (ERC) under the European 15 Union’s Horizon 2020 research and innovation programme through the ERC Consolidator Grant GAPS (N. 648913). DPG also acknowledges  support from the Spanish MINECO (project MTM2017-88385-P) and from Comunidad de Madrid (grant QUITEMAD-CM, ref. P2018/TCS-4342). APH acknowledges support from La Caixa-Severo Ochoa Grant (ICMAT Severo Ochoa project SEV-2011-0087, MINECO), and the Juan de la Cierva Grant FJC2018-036519-I.
}

\keywords{Quantum Lattice System, One-Dimensional System, Phase Transition, Locality Estimates, Decay of Correlations}

\subjclass[2010]{82B10, 82B20, 82B26}

\begin{document}

\begin{abstract}
 It is a generalized belief that there are no  thermal phase transitions in short range 1D quantum systems. However, the only known case for which this is rigorously proven is for the particular case of {\it finite range} translational invariant interactions. The proof was obtained by Araki in his seminal paper of 1969  as a consequence of pioneering locality estimates for the time-evolution operator that allowed him to prove its analiticity on the whole complex plane, when applied to a local observable. However, as for now there is no mathematical proof of the abscence of 1D thermal phase transitions if one allows exponential tails in the interactions.
In this work we extend Araki's result to include exponential (or faster) tails. Our main result is the analyticity of the time-evolution operator applied on a local observable on a suitable strip around the real line. As a consequence we obtain that thermal states in 1D exhibit exponential decay of correlations above a threshold temperature that decays to zero with the exponent of the interaction decay, recovering Araki's result as a particular case.  Our result however still leaves open the possibility of 1D thermal short range phase transitions. We conclude with an application of our result to the spectral gap problem for Projected Entangled Pair States (PEPS) on 2D lattices, via the holographic duality due to Cirac et al. 
\end{abstract}

\maketitle

\section{Introduction}

Lieb-Robinson bounds \cite{LiebRob72} are described as a non-relativistic counterpart to the finite-speed limit for the propagation of signals or perturbations in certain quantum systems. They formally establish that the dynamical evolution of a local observable under a local Hamiltonian has an approximately bounded support that grows linearly in the time variable and its velocity depends on the interaction and the underlying metric structure. For the last fifteen years, these locality bounds have been sharpened and extended, motivated by a wide range of applications including the existence of dynamics in the thermodynamic limit \cite{BraRob97, NaOgSi06}, simulation of (real-time) evolution with local Hamiltonians \cite{HaHaKoHa18}, and mainly to study features of ground states of local Hamiltonians, namely multi-dimensional Lieb-Schultz-Mattis theorems \cite{Ha04}, exponential clustering property \cite{HaKo06, NaSi06}, area law \cite{Ha07_2} or classification of phases \cite{BaMiNaSi12}, to name a few; see  \cite{Ha12, NaSi10, NaSi10_2, NaSiYo19} and references therein.

Lieb-Robinson type estimates for complex time evolution have been considered by Robinson \cite{Robinson67} and Araki \cite{Araki69} both for finite-range interactions. The latter apply to one-dimensional spin systems and establishes that the infinite-volume time evolution operator applied to a local observable is analytic in the time variable on the whole complex plane, and the support of the evolved observable is approximately bounded but grows exponentially in the modulus of the complex time variable. These estimates become particularly useful when combined with Araki-Dyson expansionals \cite{Araki73} to deal with local perturbations of equilibrium states. Applications to one-dimensional spin systems include the abscence of phase transition at every temperature for translational invariant and finite range interactions \cite{Araki69}, large deviations principles  \cite{LeRey05, Ogata10, OgataRey11} and central limit theorems \cite{Matsui02, Matsui03}. 


In this article, we aim to extend Araki's result to a wider class of interactions. In Section \ref{sec:AnalyticLocalityBounds} we collect several locality estimates for complex time evolution. The first result (Theorem \ref{Theo:localityFirstVersion}) is essentially known and dates back to Robinson \cite{Robinson67}. It applies to general lattices and interactions providing a disk around the origin where locality estimates hold. This turns out to be basically optimal for $\mathbb{Z}^{g}$ with $g \geq 2$ as shown by Bouch \cite{Bouch15}, but contrasts with Araki's result \cite{Araki69} in one-dimensional systems. To overcome this, we next provide a version in terms of the energy interaction across surfaces (Theorem \ref{Theo:analyticLRgeneral}), that applied to the case of bounded interactions in 1D leads to the main result of the paper (Theorem \ref{Theo:FundamentalEstimations}). In particular, for finite range interactions we recover Araki's result and for interactions that decay exponentially fast we get locality and analiticity on a disk around the origin whose radius scales with the exponent of the decay. We finish the section by showing that in combination with the ordinary Lieb-Robinson bounds, the previous locality estimates can be extended to a horizontal strip around the real axis of width equal to the diameter of the disk (Theorem \ref{Theo:originalLiebRobinson} and Corollary \ref{Coro:stripLiebRobinson}).

To motivate these results, we illustrate two applications. The first one is to extend Araki's result on equilibrium states by showing that, under translational invariance,  the infinite volume Gibbs state has exponential decay of correlations for every temperature greater than the inverse of the width of the aforementioned strip. In Section \ref{sec:expansionals} we introduce and present an auxiliary result on expansionals and in Section \ref{sec:PhaseTransition} we detail the argument for the abscence of phase transition along the lines of \cite{Araki69,GoNe98,Matsui01}.

The second application, and actually main motivation of this work, deals with the spectral gap problem for parent Hamiltonians of Projected Entangled Pair States (PEPS). In the recent article \cite{KaLuPe19} the authors have proved that for PEPS in 2D, if the boundary states on rectangles correspond to Gibbs states whose Hamiltonians feature nice locality properties, then the parent Hamiltonian of the PEPS is gapped. They deal with finite range interactions and leave open the case of interactions with exponential decay, that seems to concur better with numerical simulations \cite{CiPoScVe11}.  In Section \ref{sec:2DPEPS} we extend their result to this latter case. 

\subsection{Notation and terminology}\label{sec:notation}

Let $G=(\mathcal{V},\mathcal{E})$ be a finite or infinite graph with vertices $\mathcal{V}$, edges $\mathcal{E}$ and with the shortest graph distance. Usually, we will consider the $g$-dimensional lattice $\mathbb{Z}^{g}$ for some $g \in \mathbb{N}$. At each site $x \in \mathcal{V}$ consider a local finite-dimensional Hilbert space $\mathcal{H}_{x} = \mathbb{C}^{d}$, and so for each finite subset $X \subset \mathcal{V}$ we have the corresponding Hilbert space $\mathcal{H}_{X}:=\otimes_{x \in X}{\mathcal{H}_{x}}$ and algebra of observables $\mathcal{A}_{X} = \mathcal{B}(\mathcal{H}_{X})$, that is, the space of bounded and linear operators on $\mathcal{H}_{X}$. The assignment $X \mapsto \mathcal{A}_{X}$ is actually monotonic, in the sense that for two finite sets $X \subset Y \subset \mathcal{V}$ we can  identify through a canonical linear isometry
\[ \mathcal{A}_{X} \hookrightarrow \mathcal{A}_{Y} = \mathcal{A}_{X} \otimes \mathcal{A}_{Y \setminus X}\,, \quad Q \longmapsto Q \otimes \mathbbm{1}\,. \]

With this identification we  have a directed set $(\mathcal{A}_{X})_{X}$ whose direct limit is the so-called \emph{algebra of local observables} $\mathcal{A}_{loc}$. In particular, since the above inclusions are isometries,  $\mathcal{A}_{loc}$ is endowed with a natural norm becoming a normed $*$-algebra containing each $\mathcal{A}_{X}$ isometrically. The completion of $\mathcal{A}_{loc}$, denoted $\mathcal{A}_{\mathcal{V}}$ or simply $\mathcal{A}$, is the \emph{C$^\ast$-algebra of observables}. In particular, for a maybe infinite $\Lambda$ we denote by $\mathcal{A}_{\Lambda}$ the closed subspace generated by all $\mathcal{A}_{X}$ with $X \subset \Lambda$ finite, and say that $Q \in \mathcal{A}$ has support in $\Lambda$ whether $Q \in \mathcal{A}_{\Lambda}$.\\ 

Let us denote the partial trace over a finite subset $\Lambda$ by $\operatorname{Tr}_{\Lambda}: \mathcal{A}_{\mathcal{V}} \longrightarrow \mathcal{A}_{\mathcal{V} \setminus \Lambda}$ and its normalized version as $\operatorname{tr}_{\Lambda}=  \operatorname{Tr}_{\Lambda}/d^{|\Lambda|}$. The tracial state over $\mathcal{A}$ will be simply denoted as $\operatorname{tr}$.\\

We will consider a local interaction $\Phi$ on the lattice, namely a function which associates to each (non empty) finite subset $X \subset \mathcal{V}$ an element $\Phi_{X} = \Phi_{X}^{\dagger} \in \mathcal{A}_{X}$.  Let us define for each $n \geq 0$ 
\[ \Omega_{n} := \sup_{x \in 	\mathcal{V}}{ \,\, \sum{\{ \| \Phi_{X}\| \colon X \ni x\,,\, \operatorname{diam}(X) \geq n  \}}}\, \]
where $\operatorname{diam}(\cdot)$ denotes the diameter with respect to the graph distance over $G$. We will assume that  $\Omega_{0}$ is finite, condition sometimes referred as having \emph{bounded interactions}. The sequence $(\Omega_{n})$ is non-increasing and quantifies the decay of interactions between particles of the lattice. We will say that $\Phi$ has \emph{finite range} if there is $r>0$ such that $\Omega_{n}=0$ whenever $n>r$, or that has \emph{exponential decay} if there exists $\lambda > 0$ such that 
\[ \| \Phi\|_{\lambda}:= \sum_{n \geq 0}{\Omega_{n} \, e^{\lambda n}} \, < \,\infty\,. \]
\noindent The total energy corresponding to a finite subset $\Lambda \subset \mathcal{V}$ is the Hamiltonian
\[ H_{\Lambda}:= \sum_{X \subset \Lambda}{\Phi_{X}}\,. \]
The corresponding time-evolution operator in the complex variable $s \in \mathbb{C}$ is defined 
\begin{equation}\label{equa:TimeEvolutionOperator} 
\Gamma^{s}_{\Lambda}(Q) = e^{isH_{\Lambda}} Q e^{-isH_{\Lambda}}\quad, \quad Q \in \mathcal{A}
\end{equation}
or equivalently, through the Dyson series
\begin{equation}
\label{equa:TimeEvolutionPowerSeriesExpansion}
\Gamma_{H_\Lambda}^{s}(Q) = \Gamma_{\Lambda}^{s}(Q) = \sum_{m=0}^{\infty}{\frac{s^{m}}{m!} \, \delta_{H_{\Lambda}}^{m}(Q)}
\end{equation}
where $\delta_{H_{\Lambda}}(Q) := i [H_{\Lambda}, Q]$ is the commutator operator. Notice that this series converges absolutely as long as $\Lambda$ is finite. When considering the time evolution operator on the whole system (whenever it can be defined) we will simply write $\Gamma^{s}$ or $\Gamma_{H}^{s}$.\\

For the one-dimensional lattice $\mathbb{Z}$, we introduce for each $j \in \mathbb{Z}$ and denote by $\tau_{j}$ the \emph{lattice translation} homomorphism characterized for each $Q \in M_{d}(\mathbb{C})$ by 
\[ \tau_{j} \, Q^{(k)} = Q^{(j+k)}\] 
where  $Q^{(k)} \in \mathcal{A}_{\{ k\}}$ is the element that coincides with $Q$ on site $k$ and with the identity on the rest. We say that the interaction $\Phi$ is \emph{translational invariant} if for every finite subset $X \subset \mathbb{Z}$ and every $j \in \mathbb{Z}$
\[ \tau_{j} \, \Phi_{X} = \Phi_{j + X}\,. \]
This nomenclature also applies to the one-sided version $\mathcal{A}_{\mathbb{N}} = \mathcal{A}_{[1,\infty)}$. In this setup, for each $Q \in \mathcal{A}_{\mathbb{N}}$, $m \in \mathbb{N}$ and real number $x>1$ define
\begin{align*}
\| Q\|_{n} & := \inf\{ \| Q - Q_{n}\| \colon Q_{n} \in \mathcal{A}_{[1,n]} \}\,,\\[2mm]
\||Q|\|_{m,x} & := \| Q\| + \sum_{n \geq m}{\| Q\|_{n} \; x^{n}} \,.
\end{align*}
The vector subspace of $\mathcal{A}_{\mathbb{N}}$ given by
\[ \mathcal{A}_{\mathbb{N}}(x) := \{ Q \in \mathcal{A}_{\mathbb{N}} \colon \|| Q|\|_{1,x} < \infty \}\, \]
turns out to be a Banach space when endowed with any of the norms $\|| \cdot |\|_{m,x}$.\\

We will denote $\mathbb{N}_{0} = \mathbb{N} \cup \{ 0\}$. For each $n \in \mathbb{N}$ and each $\alpha = (\alpha_{1}, \ldots, \alpha_{n}) \in \mathbb{N}_{0}^{n}$ let us write $|\alpha| := \alpha_{1} + \ldots + \alpha_{n}$.

\section{Analytic Lieb-Robinson bounds}
\label{sec:AnalyticLocalityBounds}

In this section, we provide several complex-variable versions of the Lieb-Robinson bounds. Given a local observable $A$ with support in a finite set $\Lambda_{0}$, we aim to compare the time evolution of $A$ on two larger regions $\Lambda, \Lambda' \supset \Lambda_{0}$, namely the norm of the difference between $\Gamma_{\Lambda'}^{s}(A)$ and $\Gamma_{\Lambda}^{s}(A)$ in order to analize the region and rate of convergence of $s \mapsto \Gamma_{\Lambda}^{s}(A)$ as $\Lambda$ grows to cover the whole graph.

\subsection{General case}
The first result shows that locality estimates hold on a disk around the origin under certain conditions on the interaction. It adapts an older argument by Robinson \cite[Proof of Theorem 1]{Robinson67} for finite-range interactions. A similar idea is also present in the books of Ruelle \cite[Theorem 7.6.2]{Ruelle1969} and Bratteli and Robinson \cite[Theorem 6.2.4]{BraRob97} for more general interactions to argue that the Dyson series of the time-evolution operator is absolutely convergent on a neighbourhood of the origin for every local observable. However, since no locality estimates for complex variables explicitly appear in the latter we have decided to include them here. These estimates apply to a wide range of graphs and interactions.

We are going to assume that our metric graph $G = (\mathcal{V},\mathcal{E})$ is locally finite, that is, for every $m \geq 1$
\begin{equation}\label{card:function} 
\Delta(m):=\sup{\{ |X| \colon  X \subset \mathcal{V}\,, \operatorname{diam}{(X)} \leq m \}} < \infty\,. 
\end{equation}
This property holds in systems with the following regularity condition appearing in recent versions of the Lieb-Robinson estimates  \cite{NaOgSi06, NaSi10}: there is a non-increasing function \mbox{$F: [0, \infty) \longrightarrow (0, \infty)$} such that 
\[ \| F \| := \sup_{x \in \mathcal{V}} \,\, \sum_{y \in \mathcal{V}}{F(\operatorname{dist}(x,y))} < \infty. \]
Indeed, for every $m \geq 0$ we can bound
\[ \Delta(m) \leq \frac{1}{F(m)} \, \sup_{x \in \mathcal{V}} \,\, \sum \{ F(\operatorname{dist}(x,y)) \colon y \in \mathcal{V},  \operatorname{dist}(y,x) \leq m \} \leq \frac{\| F\|}{F(m)}\,. \]
We state now the result.

\begin{Theo}\label{Theo:localityFirstVersion}
Let us consider a quantum spin system over $G$ with local bounded interaction $\Phi$ satisfying for some $\lambda > 0$ 
\[ \| \Phi\|_{\lambda}^{\Delta} := \mbox{$\sum_{n \geq 0}$} \, \Omega_{n} \, e^{\lambda \Delta(n)} < \infty\,. \]
Then, for every local observable $A$ with support in $\Lambda_{0} \subset \mathcal{V}$ we have that the time evolution operator $\Gamma^{s}(A)$ is defined and analytic on the open disk centered at the origin with radius $R_{\lambda}:=\lambda/(2 \, \| \Phi\|_{\lambda}^{\Delta})$. Indeed, for every $0 \leq \ell < L$ 
\begin{equation}\label{equa:localityFirstVersion} 
\| \Gamma^{s}_{\Lambda_{L}}(A) - \Gamma_{\Lambda_{\ell}}^{s}(A) \| \, \leq \, 2 \, \| A \| \,  e^{\lambda |\Lambda_{0}|} \, R_{\lambda} \, \frac{e^{ - \ell \, \| \Phi\|_{\lambda}^{\Delta} \,  (R_{\lambda} -  |s|)}}{R_{\lambda} - |s| }\,, \quad |s| < R_{\lambda}\,
\end{equation}
where we are denoting
\[ \Lambda_{\ell}:= \{ x \in \mathcal{V} \colon \operatorname{dist}(x, \Lambda_{0}) \leq \ell \}\,. \]
\end{Theo}

\begin{proof}
Let us assume that $\| A\|=1$ and fix $0 \leq \ell < L$. Using the Dyson series expansion \eqref{equa:TimeEvolutionPowerSeriesExpansion}, we can bound
\begin{equation}\label{equa:localityFirstVersionAux1} 
\| \Gamma_{\Lambda_{L}}^{s}(A) - \Gamma_{\Lambda_{\ell}}^{s}(A) \| \leq \sum_{m=1}^{\infty} \, \frac{|s|^{m}}{m!} \| \delta_{H_{\Lambda_{L}}}^{m}(A) - \delta_{H_{\Lambda_{\ell}}}^{m}(A) \|\,. 
\end{equation}
The $m$-th iterated commutator $\delta_{H_{\Lambda}}^{m}(A)$ for a finite region $\Lambda$ can be expanded as
\begin{equation}\label{equa:commutatorIteratedDecomp} 
\delta_{H_{\Lambda}}^{m} (A) = \sum_{ X_{1} \cap S_{0} \neq \emptyset \, \ldots \, , \, X_{m} \cap S_{m-1} \neq \emptyset}{  \delta_{\Phi_{X_{m}}}\circ \ldots \circ \delta_{\Phi_{X_{1}}} (A)}  
\end{equation}
where we are denoting $S_{0} := \Lambda_{0}$ and $S_{j} := X_{j} \cup S_{j-1}$ for each $j \geq 1$, and the sum is extended over subsets $X_{j} \subset \Lambda$. Hence, 
\begin{align}\label{equa:localityFirstVersionAux2} 
\delta_{H_{\Lambda_{L}}}^{m} (A) - \delta_{H_{\Lambda_{\ell}}}^{m} (A) = \sum_{\substack{X_{1} \cap S_{0} \neq \emptyset \, , \, \ldots \,, \, X_{m} \cap S_{m-1} \neq \emptyset \\[1mm] S_{m} \cap (V \setminus \Lambda_{\ell}) \neq \emptyset}}   \delta_{\Phi_{X_{n}}}\circ \ldots \circ \delta_{\Phi_{X_{1}}} (A)  
\end{align}
where the sum extends over subsets $X_{j} \subset \Lambda_{L}$, and the condition \mbox{$S_{m} \cap (V \setminus \Lambda_{\ell}) \neq \emptyset$} follows from the fact that summands corresponding to $X_{1}, \ldots, X_{m} \subset \Lambda_{\ell}$ vanish when substracting. For each $j \geq 1$, rewrite
\[ \sum_{X_{j} \cap S_{j-1} \neq \emptyset} = \sum_{\alpha_{j} \in \mathbb{N}_{0}} \sum_{\substack{X_{j} \cap S_{j-1} \neq\emptyset \\[1mm] \operatorname{diam}(X_{j}) = \alpha_{j}}} \]
Notice that the additional condition $S_{m} \cap (V \setminus \Lambda_{\ell}) \neq \emptyset$ forces the sum of the diameters of all $X_{j}$ to be greater than or equal to $\ell$. We can then bound  \eqref{equa:localityFirstVersionAux2} 
\begin{align*}
\| \delta_{H_{\Lambda_L}}^{m}(A)  - \delta_{H_{\Lambda_\ell}}^{m}(A)\|  \leq 2^{m}  \, \sum_{\substack{\alpha \in \mathbb{N}_{0}^{m}\\ |\alpha| \geq \ell}} \sum_{\substack{X_{1} \cap S_{0} \neq \emptyset \\ \operatorname{diam}(X_{1}) = \alpha_{1}}} \, \ldots \, \sum_{\substack{X_{m} \cap S_{m-1} \neq \emptyset \\ \operatorname{diam}{(X_{m})} = \alpha_{m}}} \,\,\, \prod_{j=1}^{m}{\| \Phi_{X_{j}}\|}\,.
\end{align*}
Let us fix $\varepsilon \in (0, \lambda)$. Using that $|S_{j}| \leq |\Lambda_{0}| + |X_{1}| + \ldots + |X_{j}|$ 
\begin{align*} 
\| \delta_{H_{\Lambda_L}}^{m}(A)  - \delta_{H_{\Lambda_\ell}}^{m}(A)\| & \leq 2^{m}  \, \sum_{\substack{\alpha \in \mathbb{N}_{0}^{m}\\ |\alpha| \geq \ell}} \,\, \prod_{j=1}^{m}  \,\, \left( |\Lambda_{0}| + \Delta(\alpha_{1}) + \ldots + \Delta(\alpha_{j-1}) \right) \, \Omega_{\alpha_{j}} \\[2mm]
& \leq 2^{m}  \, \sum_{\substack{\alpha \in \mathbb{N}_{0}^{m}\\ |\alpha| \geq \ell}} \, \left(|\Lambda_{0}| + \Delta(\alpha_{1}) + \ldots + \Delta(\alpha_{m})\right)^{m} \, \,\, \prod_{j=1}^{m} \,\, \Omega_{\alpha_{j}}\\[2mm]
& \leq 2^{m}  \, \sum_{\substack{\alpha \in \mathbb{N}_{0}^{m}\\ |\alpha| \geq \ell}} \, \frac{m!}{\varepsilon^{m}} e^{ \varepsilon \left(|\Lambda_{0}| + \Delta(\alpha_{1}) + \ldots + \Delta(\alpha_{m})\right)} \, \,\, \prod_{j=1}^{m} \,\, \Omega_{\alpha_{j}}\\[2mm]
& \leq \left( 2/\varepsilon\right)^{m} \,  m! \, e^{\varepsilon |\Lambda_{0}|}  \, \sum_{\substack{\alpha \in \mathbb{N}_{0}^{m}\\ |\alpha| \geq \ell}} \, \,\, \prod_{j=1}^{m} \, \, \Omega_{\alpha_{j}} e^{\varepsilon \Delta(\alpha_{j})}  \,.
\end{align*}
Since $\Delta(n) \geq 1$ for every $n \geq 0$ 
\begin{align*}
\| \delta_{H_{\Lambda_L}}^{m}(A)  - \delta_{H_{\Lambda_\ell}}^{m}(A)\| & \leq \left( 2/\varepsilon\right)^{m} \, m! \, e^{\varepsilon |\Lambda_{0}|} \, e^{\ell \, (\varepsilon - \lambda)}  \, \sum_{\substack{\alpha \in \mathbb{N}_{0}^{m}\\ |\alpha| \geq \ell}} \, \,\, \prod_{j=1}^{m} \, \, \Omega_{\alpha_{j}} e^{\lambda \Delta(\alpha_{j})} \\[2mm]
& \leq \left( 2/\varepsilon\right)^{m}  \, m! \, e^{\varepsilon |\Lambda_{0}|} \, e^{\ell \, (\varepsilon - \lambda)}  \, (\| \Phi\|_{\lambda}^{\Delta})^{m} \,.
\end{align*}
Finally, applying this estimate to \eqref{equa:localityFirstVersionAux1} we deduce that for every $s \in \mathbb{C}$ with $2  \, |s| \, \| \Phi\|_{\lambda}^{\Delta}  < \lambda $ and taking 
\[ \varepsilon:= \frac{\lambda + 2 \, |s| \, \| \Phi\|_{\lambda}^{\Delta}}{2} = \lambda + \| \Phi\|_{\lambda}^{\Delta} \, (|s| - R_{\lambda}) \] 
which satisfies $2 |s| \| \Phi\|_{\lambda}^{\Delta} < \varepsilon < \lambda$
\begin{align*} 
\| \Gamma^{s}_{\Lambda_{L}}(A) - \Gamma^{s}_{\Lambda_{\ell}}(A) \| \, & \leq  \, e^{\ell (\varepsilon - \lambda)} \, e^{\varepsilon |\Lambda_{0}|} \,  \, \sum_{m=1}^{\infty}{ \left(\frac{2 |s| \, \| \Phi\|_{\lambda}^{\Delta}}{\varepsilon}\right)^{m} }\\[2mm] 
& \leq  \, e^{\ell (\varepsilon - \lambda)} \, e^{\varepsilon |\Lambda_{0}|} \,  \frac{\varepsilon}{\varepsilon - 2 |s| \| \Phi\|_{\lambda}^{\Delta}} \\[2mm]
& \leq \, e^{\ell(\varepsilon - \lambda)} \, e^{\lambda |\Delta_{0}|} \, \frac{\lambda}{\lambda - \varepsilon} \\[2mm]
& \leq e^{\ell \, \| \Phi\|_{\lambda}^{\Delta}(|s| - R_{\lambda})} \, e^{\lambda \, \Delta_{0}} \, \frac{\lambda/\| \Phi\|_{\lambda}^{\Delta}}{  R_{\lambda} - |s|}\,.
\end{align*}
and the desired estimate immediately follows.
\end{proof}

If $G = \mathbb{Z}^{g}$ with $g \geq 2$ then $ \Delta(m) \leq (2g)^{m}$, so for sufficiently fast decay interactions, e.g. finite range, the previous theorem guarantees analyticity and locality estimates in a suitable disk around the origin. This can be considered tight, since according to \cite{Bouch15} we cannot expect convergence in the whole complex plane, not even for $g=2$ and translation invariant nearest neighbor interactions.\\

If $G= \mathbb{Z}$ then $\Delta(m) = m+1$ and so the theorem applies to interactions with exponential decay. However, the result seems not tight. Indeed,  for finite range interactions we can choose any value $\lambda > 0$ and the previous theorem only provides a disk of convergence of order $\lambda/e^{\lambda}$. This contrasts with Araki's result \cite{Araki69} which ensures convergence on the whole complex plane. We will overcome this drawback in the following subsections.

\subsection{Energy across surface}

The next locality estimates involve properties of the energy interaction across surfaces, see \cite[p. 249-51]{BraRob97}. Notice that the statement does not depend on the graph distance, and so may be applicable to a wider variety of situations.

\begin{Theo}\label{Theo:analyticLRgeneral}
Let us consider a quantum spin system over $G$ with interaction $\Phi$, and fix an increasing sequence $(\Lambda_{n})_{n \geq 0}$ of finite subsets of $\mathcal{V}$. For each  $0 \leq j < k$ let us denote
\begin{align*}
W(j,k) & := \sum{\{ \| \Phi_{X}\| \colon  X \cap \Lambda_{j} \neq \emptyset \,,\, X \cap (\Lambda_{k} \setminus \Lambda_{k-1}) \neq \emptyset \}}\,,\\[2mm]  
W(j,j) & := \sum{\{\| \Phi_{X} \|\colon X \cap \Lambda_{j} \neq \emptyset \}} \leq |\Lambda_{j}| \, \Omega_{0} \,. 
\end{align*}
Then, for every local observable $A \in \mathcal{A}$ with support in $\Lambda_{0}$ and every $0 \leq \ell \leq L$  
\begin{align}
\label{equa:analyticLRgeneral0} \| \Gamma_{\Lambda_{L}}^{s}(A)\| \, & \leq \, \| A\| \, \sum_{k=0}^{L}{e^{2 |s| \Omega_{0} |\Lambda_{k}|} \,  W^{\ast}_{k}(2|s|)}\\[2mm]  
\|\Gamma_{\Lambda_{L}}^{s}(A) - \label{equa:analyticLRgeneral1}  \Gamma_{\Lambda_{\ell}}^{s}(A)\| \, & \leq \, \| A\| \, \sum_{k=\ell+1}^{L}{e^{2 |s| \Omega_{0} |\Lambda_{k}|} \,  W^{\ast}_{k}(2|s|)} 
\end{align}
where
\[ W^{\ast}_{0}(x) = 1 \quad, \quad W^{\ast}_{k}(x) := \sum_{n=1}^{\infty} \left( \sum_{0 = \beta_{0} < \ldots < \beta_{n}=k} \,\, \prod_{j=1}^{n} W(\beta_{j-1}, \beta_{j}) \right)\frac{x^{n}}{n!}\,. \]
\end{Theo}

\begin{proof}
Let us assume that $\| A\| = 1$. We just have to prove that \eqref{equa:analyticLRgeneral1} holds for consecutive regions $\Lambda_{k}$ and $\Lambda_{k-1}$, as the general case follows straightforwardly through a telescopic sum of terms. We again make use of the estimate
\begin{equation}\label{equa:analyticLRgeneralAux1}
\| \Gamma_{\Lambda_{k}}^{s}(A) - \Gamma_{\Lambda_{k-1}}^{s}(A) \| \leq \sum_{m=1}^{\infty} \, \frac{|s|^{m}}{m!} \| \delta_{H_{\Lambda_{k}}}^{m}(A) - \delta_{H_{\Lambda_{k-1}}}^{m}(A) \|\,. 
\end{equation}
We must then find good estimates of the summands in the right-hand side of \eqref{equa:analyticLRgeneralAux1}. The argument is split into several stages.\\

\noindent \textbf{Step I}:  \emph{Let us denote for every $m \in \mathbb{N}$ 
\[ \mathcal{U}_{0}^{(m)}:= \delta_{H_{\Lambda_{0}}}^{m} (A) \quad , \quad \mathcal{U}_{k}^{(m)}:= \delta_{H_{\Lambda_{k}}}^{m} (A) - \delta_{H_{\Lambda_{k-1}}}^{m}(A) \,\,\quad (k > 0) \,. \]
Then, for every $k \geq 0$
\begin{align}\label{equa:boundRecursiveFormula}
 \quad \mathcal{U}_{k}^{(m+1)}  = \delta_{H_{\Lambda_{k}}} ( \mathcal{U}_{k}^{(m)} ) \,+ \, \sum_{i=0}^{k-1}  \, \delta_{H_{\Lambda_{k}} - H_{\Lambda_{k-1}}}  ( \mathcal{U}_{i}^{(m)})
\end{align}
where for $k=0$ the finite series on the right-hand side of \eqref{equa:boundRecursiveFormula} is equal to zero}. To prove this, note that fixed $m \geq 1$ and $k \geq 0$ we can decompose
\[ \delta_{H_{\Lambda_{k}}}^{m}(A) = \sum_{i=0}^{k}{\mathcal{U}_{i}^{(m)}}\,. \]
Applying the operator $\delta_{H_{\Lambda_{k}}}$ on both sides we get 
\begin{align*}
\delta_{H_{\Lambda_{k}}}^{m+1} (A) \,  & = \, \sum_{i=0}^{k} \left( \delta_{H_{\Lambda_{i}}} \, (\mathcal{U}_{i}^{(m)}) + \sum_{j=i+1}^{k}{\delta_{H_{\Lambda_{j}} - H_{\Lambda_{j-1}}} \, (\mathcal{U}_{i}^{(m)}}) \right)\\[2mm]
& = \, \sum_{i=0}^{k}{\delta_{H_{\Lambda_{i}}} \, (\mathcal{U}_{i}^{(m)}}) + \sum_{i=0}^{k} \sum_{j=i+1}^{k} \delta_{H_{\Lambda_{j}} - H_{\Lambda_{j-1}}} \, (\mathcal{U}_{i}^{(m)})\\[2mm]
& = \, \sum_{i=0}^{k}{\delta_{H_{\Lambda_{i}}} \, (\mathcal{U}_{i}^{(m)}}) + \sum_{j=1}^{k} \sum_{i=0}^{j-1} \delta_{H_{\Lambda_{j}} - H_{\Lambda_{j-1}}} \, (\mathcal{U}_{i}^{(m)})\\[2mm]
& = \, \sum_{j=0}^{k} \left( \delta_{H_{\Lambda_{j}}} \, (\mathcal{U}_{j}^{(m)}) + \sum_{i=0}^{j-1} \delta_{H_{\Lambda_{j}} - H_{\Lambda_{j-1}}} \, (\mathcal{U}_{i}^{(m)}) \right)\,.
\end{align*}
Using this identity, we can immediately check that
\[ \delta_{H_{\Lambda_{k}}}^{m+1} (A) - \delta_{H_{\Lambda_{k-1}}}^{m+1}(A)\, = \, \delta_{H_{\Lambda_{k}}} \, (\mathcal{U}_{k}^{(m)}) + \sum_{i=0}^{k-1}{\delta_{H_{\Lambda_{k}} - H_{\Lambda_{k-1}}} \, (\mathcal{U}_{i}^{(m)}})\,, \]
which finishes the proof of the statement.\\

\noindent \textbf{Step II}:  \emph{Define for every $m \geq 1$ and $k \geq 0$ the nonnegative number
\begin{align*}  
\phi(m,k):= \sum_{ 0 = \beta_{0} \leq \ldots \leq \beta_{m} = k}{ \,\,\, \prod_{j=1}^{m}{ W(\beta_{j-1}, \beta_{j}) }}\,.
\end{align*}
Then, we have
\begin{align} \label{equa:firstEstimateDerivations}
\| \mathcal{U}_{k}^{(m)} \| \, \leq \, 2^{m}  \, \phi(m,k)\,. 
\end{align}}Note that  if $0 \leq j < k$ and $B$ is an observable with support in $\Lambda_{j}$, then 
\begin{equation}
\| \delta_{H_{\Lambda_{j}}}  (B)\|  \, \leq \,  \sum{\{ \| \delta_{\Phi_{X}} (B) \| \colon X \cap \Lambda_{j} \neq \emptyset \}} \, \leq \, 2 \, \| B\| \, W(j,j)
\end{equation}
and 
\begin{equation}
\begin{split}
\| \delta_{H_{\Lambda_{k}} - H_{\Lambda_{k-1}}} (B) \| \, & \leq \, \sum{\{ \| \delta_{\Phi_{X}}(B) \| \colon X \cap \Lambda_{j} \neq \emptyset\,, X \subset \Lambda_{k} \,, X \nsubseteq \Lambda_{k-1}  \}}\\[2mm]
& \leq \, 2 \, \| B\|\, W(j,k) \,.
\end{split}
\end{equation}
These inequalities will be helpful to prove that \eqref{equa:firstEstimateDerivations} holds for every $k \geq 0$ by induction on $m \geq 1$. Indeed, the case $m=1$ follows taking $j = 0$ and $B = A$. Fixed $m \geq 1$ and assuming that formula \eqref{equa:firstEstimateDerivations} holds for every $k \geq 0$, we can apply the recursive formula \eqref{equa:boundRecursiveFormula}  from Step I to get  for every $k \geq 0$ 
\begin{align*}
\| \mathcal{U}_{k}^{(m+1)} \| \, & \leq \, \|\delta_{H_{\Lambda_{k}}} (\mathcal{U}_{k}^{(m)}) \| + \sum_{i=0}^{k-1} \| \delta_{H_{\Lambda_{k}} - H_{\Lambda_{k-1}}} (\mathcal{U}_{i}^{(m)})\|\\[2mm]
& \leq \, 2 \, \| \mathcal{U}_{k}^{(m)}\| \, W(k,k) \, + \,  \sum_{i=0}^{k-1} 2 \, \| \mathcal{U}_{i}^{(m)}\| \, W(i,k) \\[2mm]
& \leq \, 2^{m+1}  \, \sum_{i=0}^{k}  \phi(m,i) \, W(i, k)\\[2mm] 
& = \, 2^{m+1}  \, \phi(m+1, k) \,.
\end{align*}
This finishes the proof of the statement.\\

\noindent \textbf{Step III}: We are going to bound $\phi(m,k)$ by rearranging terms $W(\beta_{j-1}, \beta_{j})$ into those with $\beta_{j-1} \neq \beta_{j}$ and those with $\beta_{j-1} = \beta_{j}$. Note that for $0 \leq j \leq k$, we can bound $W(j,j) \leq W(k,k) \leq \Omega_{0} \, |\Lambda_{k}|$. Moreover, if $n,m \geq 1$ and we fix a sequence of integers \mbox{$0 = \alpha_{0} < \ldots < \alpha_{n} = k$}, then the number of sequences of integers \mbox{$0 = \beta_{0} \leq  \ldots \leq \beta_{m} = k$} with $\{ \beta_{0}, \ldots, \beta_{m}\} = \{ \alpha_{0}, \ldots, \alpha_{n} \}$ coincides with the number of subsets with cardinality $n$ of $\{ 1, \ldots, m\}$. Therefore, 
\begin{equation}\label{equa:mainBoundPhi} 
\phi(m,k) \leq \sum_{n=1}^{m}{\left( \sum_{0 = \beta_{0} < \ldots < \beta_{n} = k} \,\, \prod_{j=1}^{n}{W(\beta_{j-1}, \beta_{j})} \right) \Omega_{0}^{m-n} \, |\Lambda_{k}|^{m-n} \, \binom{m}{n}}\,. 
\end{equation}

\noindent \textbf{Step IV}: Applying estimates \eqref{equa:firstEstimateDerivations} and \eqref{equa:mainBoundPhi} to \eqref{equa:analyticLRgeneralAux1}, we conclude
\begin{align*}
\| \Gamma_{\Lambda_{k}}^{s}&(A)  -  \Gamma_{\Lambda_{k-1}}^{s}(A) \| \, \, \leq \\[2mm] 
& \, \leq  \, \sum_{m=1}^{\infty} \frac{(2 |s|)^{m}}{m!} \, \sum_{n=1}^{m}{\left( \sum_{0 = \beta_{0} < \ldots < \beta_{n} = k} \,\, \prod_{j=1}^{n}{W(\beta_{j-1}, \beta_{j})} \right) \Omega_{0}^{m-n} \, |\Lambda_{k}|^{m-n} \, \binom{m}{n}}\\[2mm]
& = \sum_{m=1}^{\infty} \sum_{n=1}^{m} \left( \sum_{0 = \beta_{0} < \ldots < \beta_{n} = k} \,\, \prod_{j=1}^{n}{W(\beta_{j-1}, \beta_{j})} \right) \, \frac{\Omega_{0}^{m-n} \, |\Lambda_{k}|^{m-n} (2 |s|)^{m-n}}{(m-n)!} \, \frac{(2 |s|)^{n}}{n!}\\[2mm]
& \leq \sum_{n=1}^{\infty} \frac{(2 |s|)^{n}}{n!} \, \left( \sum_{0 = \beta_{0} < \ldots < \beta_{n} = k} \,\, \prod_{j=1}^{n}{W(\beta_{j-1}, \beta_{j})} \right) \, \sum_{m=n}^{\infty} \frac{\Omega_{0}^{m-n} \, |\Lambda_{k}|^{m-n} (2 |s|)^{m-n}}{(m-n)!}\\[2mm]
& = e^{2 |s| \Omega_{0} |\Lambda_{k}|} \, \sum_{n=1}^{\infty} \frac{(2 |s|)^{n}}{n!} \, \left( \sum_{0 = \beta_{0} < \ldots < \beta_{n} = k} \,\, \prod_{j=1}^{n}{W(\beta_{j-1}, \beta_{j})} \right)\,.
\end{align*}
Finally, estimate \eqref{equa:analyticLRgeneral0} also follows  using telescopic sums and the fact that
\[ \|\Gamma^{s}_{\Lambda_{0}}(A)\| \leq \sum_{m=0}^{\infty} \frac{|s|^{m}}{m!} \| \mathcal{U}_{0}^{(m)}\| \, \leq \, \sum_{m=0}^{\infty} \frac{(2|s|)^{m}}{m!}(\Omega_{0} |\Lambda_{0}| )^{m} \leq e^{2|s| \Omega_{0}|\Lambda_{0}|}\,. \]
\end{proof}

In contrast to Theorem \ref{Theo:localityFirstVersion}, the previous result is not suitable for $\mathbb{Z}^{g}$ with $g\geq 2$.  Indeed, consider a quantum spin system over $\mathbb{Z}^{2}$, an observable $A$ supported at the origin, the sequence $\Lambda_{n} = [-n,n]^{2}$ and some nearest neighbour interaction with $W(j, j + 1) > \varepsilon > 0$ for every $j \geq 0$ and some $\varepsilon > 0$. Then,
 $W_{k}^{\ast}(2|s|) \geq (2\varepsilon |s|)^{k}/k!$ and  $e^{2 |s| \Omega_{0} |\Lambda_{k}|} = e^{2|s|\Omega_{0} (2k+1)^{2}}$, so the series in \eqref{equa:analyticLRgeneral1} cannot converge as $L$ tends to infinity unless $s$ is equal to zero.

\subsection{One-dimensional lattice}

We are going to particularize Theorem \ref{Theo:analyticLRgeneral} to the one-dimensional lattice $\mathbb{Z}$. The natural supporting sets to consider here are intervals $J=[a,b]$, for which $J_{k}=[a-k,b+k]$ for every $k \in \mathbb{N}$ so that $|J_{k} \setminus J_{k-1}| \leq 2$. Following the notation of the  aforementioned Theorem \ref{Theo:analyticLRgeneral}, 
\[ W(j,k) \leq 2 \, \Omega_{k-j} \quad \mbox{for every $0 \leq j < k$}\,. \]
This immediately yields the following main result.
\begin{Theo}\label{Theo:FundamentalEstimations}
Let us consider a quantum spin system over $\mathbb{Z}$ with bounded interaction $\Phi$ and let $A \in \mathcal{A}_{loc}$ be a local observable with support in an interval $J$. Then, for every $0 \leq \ell \leq L$ and every $s \in \mathbb{C}$ 
\begin{align}
\label{equa:FundamentalEstimations1} \| \Gamma^{s}_{J_{L}}(A) - \Gamma^{s}_{J_{\ell}}(A)   \|  & \,\, \leq \,\, e^{2 |s| \Omega_{0} |J|} \, \sum_{k=\ell + 1}^{L} e^{4 |s| \Omega_{0} k}\, \Omega_{k}^{*}(4 |s|)\\
\label{equa:FundamentalEstimations2}  \| \Gamma^{s}_{J_{L}}(A) \| & \,\, \leq \,\, e^{2 |s| \Omega_{0} |J|} \,\,\,  \sum_{k=0}^{L} e^{4 |s| \Omega_{0} k}\, \Omega_{k}^{*}(4 |s|)
 \end{align}
 where 
\begin{equation}\label{equa:OmegaStarExpression1} 
\Omega_{0}^{\ast}(x) = 1\,, \quad  \Omega_{k}^{*}(x) = \sum_{n=1}^{\infty}{ \Big( \sum_{\substack{\alpha \in \mathbb{N}^{n}\\ |\alpha| = k}} \, \prod_{j=1}^{n}\Omega_{\alpha_{j}} \Big) \frac{x^{n}}{n!} } \,.
\end{equation}
Notice that the $\Omega_{k}^{*}(x)$'s are the coefficients in the power series expansion of
\begin{equation}\label{equa:OmegaStarExpression2} 
\mathbb{D} \ni z \, \longmapsto \, \exp{\left[ x \, \sum_{k=1}^{\infty} \Omega_{k} z^{k}\right]} \, = \, \sum_{k = 0}^{\infty}{\Omega_{k}^{\ast}(x) \, z^{k}} \,. 
\end{equation}
\end{Theo}

In the rest of the section, we shed more light on the bounds in the previous theorem by particularizing them to interactions with finite range and exponential decay. Next, we reformulate the usual Lieb-Robinson bounds for real time evolution in terms of the sequence $\Omega_{n}$ to compare it with the previous result, and finally present a combined locality estimate for real and imaginary time evolution.

\subsubsection{Finite-range interactions}

Assume that there exists $r >1$ such that $\Omega_{k} = 0$ if $k \geq r$. We estimate $\Omega_{k}^{\ast}(x)$ using the explicit expression \eqref{equa:OmegaStarExpression1}. Note that if $\alpha \in \mathbb{N}^{n}$ satisfies $|\alpha| = k \geq r n$ then  $\alpha_{j} \geq r$ for some $j$, and so \mbox{$\Omega_{\alpha_{1}} \cdot \ldots \cdot \Omega_{\alpha_{n}} =0$}. Moreover, the set of elements $\alpha \in \mathbb{N}^{n}$ with $|\alpha|=k$ has cardinal $\binom{k-1}{n-1} = \frac{n}{k} \binom{k}{n} \leq  \left( \frac{k e}{n}\right)^{n} \leq (re)^{n}$ where in the last inequality we have used the constraint $k<rn$. Thus
\begin{align*}
\sum_{k= \ell + 1}^{L}{e^{x \Omega_{0} k} \, \Omega^{\ast}_{k}(x)} & \leq \sum_{k=\ell + 1}^{L} e^{x \Omega_{0} k}  \, \sum_{k/r< \, n \, \leq k}  \frac{(x\Omega_{0}re)^{n}}{n!} \leq \sum_{n>\ell/r} nr e^{x \Omega_{0}nr} \frac{(x\Omega_{0}re)^{n}}{n!}\\[2mm]
& \leq \sum_{n > \ell/r} \left( x \Omega_{0} r^{2} e^{1+x \Omega_{0} r} \right)^{n}\frac{1}{n!}\\
& \leq \exp{\left(x \Omega_{0} r^{2} e^{1+x\Omega_{0} r} \right)} \, \frac{\left(x \Omega_{0} r^{2} e^{1+x\Omega_{0} r} \right)^{[\ell/r] + 1}}{\left([\ell/r]+1 \right)!}
\end{align*}

This estimate shows that for every observable $A$ with support in a compact interval $J$, the sequence of maps $s \mapsto \Gamma_{J_{\ell}}^{s}(A)$ converges superexponentially fast as $\ell$ tends to infinity on every bounded subset of the complex plane. Indeed, the asymptotic order of convergence coincides with the one provided by Araki in \cite[Theorem 4.2]{Araki69}.

\subsubsection{Exponentially decaying interactions}

Let us assume that our interaction  has exponential decay $\| \Phi\|_{\lambda} < \infty$ for some $\lambda > 0$. Then, for every $x \geq 0$
\[ \Omega_{k}^{\ast}(x) = e^{-\lambda k} \, \sum_{n=1}^{\infty} \Big( \sum_{\substack{\alpha \in \mathbb{N}^{n}\\ |\alpha| = k}} \, \prod_{j=1}^{n}{\Omega_{\alpha_{j}} e^{\lambda \alpha_{j}}} \Big)\frac{x^{n}}{n!}  \]
and so for $0 \leq x \leq \frac{\lambda}{\Omega_{0}}$ we can bound
\begin{align*} 
\sum_{k=\ell + 1}^{L}{e^{x \Omega_{0} k} \, \Omega_{k}^{\ast}(x)} & \leq e^{(x \Omega_{0} - \lambda) \ell} \sum_{k=\ell+1}^{L} \sum_{n=1}^{\infty} \Big( \sum_{\substack{\alpha \in \mathbb{N}^{n}\\ |\alpha| = k}} \, \prod_{j=1}^{n}{\Omega_{\alpha_{j}} e^{\lambda \alpha_{j}}} \Big) \, \frac{x^{n}}{n!} \\[2mm] 
& \leq e^{(x \Omega_{0} - \lambda)\ell} \, e^{ x \| \Phi\|_{\lambda} }  \,.
\end{align*}
Combining this estimate with Theorem \ref{Theo:FundamentalEstimations} we conclude that  for every local observable $A$ with support in an interval $J$ and every $s \in \mathbb{C}$ with $|s| \leq \frac{\lambda}{4 \Omega_{0}}$
\begin{align}
\label{equa:localityEstimatesExponentialDecay1} \| \Gamma^{s}_{J_{L}}(A) - \Gamma^{s}_{J_{\ell}}(A)   \|  & \,\, \leq \,\, e^{2|s|\Omega_{0}|J|}  \,\, e^{ 4|s| \| \Phi\|_{\lambda} }\,\, e^{(4|s|\Omega_{0} - \lambda)  \ell}\,, \\[2mm]
\label{equa:localityEstimatesExponentialDecay2} \| \Gamma^{s}_{J_{L}}(A) \| & \,\, \leq \,\, e^{2|s|\Omega_{0}|J|} \,\, e^{ 4|s| \| \Phi\|_{\lambda} }\,.
 \end{align}
As a consequence, the sequence $s \mapsto \Gamma^{s}_{J_{\ell}}(A)$ converges as $\ell$ tends to infinite  on every compact subset of the open disk centered at the origin and radius $\lambda/(4 \Omega_{0})$ exponentially fast. 

\subsubsection{Lieb-Robinson bounds}

Let us illustrate how the Lieb-Robinson bounds for real time evolution can be reformulated in terms of the sequence $\Omega_{k}$. They are significantly better than the complex time version that we obtained for the one-dimensional case in Theorem \ref{Theo:FundamentalEstimations}.

\begin{Theo}\label{Theo:originalLiebRobinson}
Let us consider a quantum spin system over $\mathbb{Z}$ with local bounded interaction $\Phi$  and let  $A$ be a local observable with support in a finite set $\Lambda_{0} \subset \mathbb{Z}$. Then, for every $0 \leq \ell \leq L$ and every $t \in \mathbb{R}$ we have
\begin{align*} 
\| \Gamma_{\Lambda_{L}}^{t}(A) - \Gamma_{\Lambda_{\ell}}^{t}(A) \| \, \leq \, \| A\| \, |\Lambda_{0}|\, e^{4 |t| \Omega_{0}} \, \sum_{k \geq \ell + 1} \, \Omega_{k}^{\ast}(4 |t|)\,. 
\end{align*}
where we are denoting $\Lambda_{\ell} := \{ x \in \mathbb{Z} \colon \operatorname{dist}(x, \Lambda_{0}) \leq \ell \}$.
\end{Theo}

\begin{proof}
Let us recall that by \cite[Section 2.1]{NaOgSi06}, for every finite subset $\Lambda \subset \mathbb{Z}$, every local observable $B$ with support in $\mathfrak{B}$ and every $t \in \mathbb{R}$
\[ \| [\Gamma^{t}_{\Lambda}(A), B] \| \leq \| A\| \, \| B\| \, \sum_{k=0}^{\infty}{\frac{(2|t|)^{k}}{k!} \, a_{k}}\,, \] 
where $a_{0} = 1$ and for each $k \geq 1$
\begin{equation*}
a_{k} = \sum_{\substack{ X_{1} \cap \Lambda_{0} \neq \emptyset}} \,\, \sum_{  X_{2} \cap X_{1} \neq \emptyset} \,\, \ldots \, \sum_{ \substack{X_{k} \cap X_{k-1} \neq \emptyset \\ X_{k} \cap \mathfrak{B} \neq \emptyset}} \,\, \prod_{j=1}^{k}\| \Phi_{X_{j}}\| \,
\end{equation*}
where the sum is extended over subsets $X_{j} \subset \Lambda$. Combining this  with the following inequality (see e.g. \cite[Section 2.2]{NaOgSi06})
\[ \| \Gamma_{\Lambda_{L}}^{t}(A) - \Gamma_{\Lambda_{\ell}}^{t}(A) \| \,  \leq \, \sum_{\substack{ X \cap \Lambda_{\ell}^{c} \neq \emptyset \\ X \subset \Lambda_{L}}} \,\, \int_{0}^{|t|}{\|  [ \Gamma_{\Lambda_{\ell}}^{u}(A) \,,  \Phi_{X} ]   \|}\, d u\,, \]
we can estimate
\begin{equation} 
\| \Gamma_{\Lambda_{L}}^{t}(A) - \Gamma_{\Lambda_{\ell}}^{t}(A) \| \, \leq \, \|A \| \, \sum_{k=1}^{\infty}{\frac{(2 |t|)^{k}}{k!} b_{k}} 
\end{equation}
where
\begin{align*}
b_{k} \, & = \sum_{X_{1} \cap \Lambda_{0} \neq \emptyset} \,\, \sum_{ X_{2} \cap X_{1} \neq \emptyset} \,\, \ldots \, \sum_{X_{k-1} \cap X_{k-2} \neq \emptyset} \,\, \sum_{\substack{ X_{k} \cap X_{k-1} \neq \emptyset \\ X_{k} \cap \Lambda_{\ell}^{c} \neq \emptyset}} \,\, \| \Phi_{X_{1}}\| \, \ldots  \, \| \Phi_{X_{k}}\|  \\
& \leq \, \sum_{j_{0} \in \Lambda_{0}} \,\, \sum_{j_{1}, \ldots, j_{k-1} \in \mathbb{Z}} \,\, \sum_{j_{k} \in \Lambda_{\ell}^{c}} \,\, \sum_{X_{1} \ni j_{0}, j_{1}} \, \ldots \sum_{X_{k} \ni j_{k-1}, j_{k}} \,\,  \| \Phi_{X_{1}}\| \, \ldots \, \| \Phi_{X_{k}}\|\\
& \leq \, \sum_{j_{0} \in \Lambda_{0}} \,\, \sum_{j_{1}, \ldots, j_{k-1} \in \mathbb{Z}} \,\, \sum_{j_{k} \in \Lambda_{\ell}^{c}} \,\, \Omega_{|j_{1} - j_{0}|} \cdot \ldots \cdot \Omega_{|j_{k} - j_{k-1}|}\\[1mm]
&  \leq 2^{k} \, |\Lambda_{0}| \, \sum_{\alpha \in \mathbb{N}_{0}^{k}, \, |\alpha| > \ell} \Omega_{\alpha_{1}} \cdot \ldots \cdot \Omega_{\alpha_{k}}\\
& \leq 2^{k} \, |\Lambda_{0}| \, \sum_{j=1}^{k} \, \Big( \sum_{\alpha \in \mathbb{N}^{j}, \, |\alpha| > \ell} \Omega_{\alpha_{1}} \cdot \ldots \cdot \Omega_{\alpha_{j}} \Big) \, \Omega_{0}^{k-j} \, \binom{k}{j}\,.
\end{align*}
Arguing now as in Step IV of the proof of Theorem \ref{Theo:analyticLRgeneral},  we conclude the result
\end{proof}

Finally, we present a locality estimate which combines the complex time version with the Lieb-Robinson bounds. 

\begin{Coro} \label{Coro:stripLiebRobinson}
Let us consider a quantum spin system over $\mathbb{Z}$ with interaction $\Phi$ having exponential decay $\| \Phi\|_{\lambda} < \infty$ for some $\lambda>0$. Then, for very local observable $A$ with support in an interval $J$, every $1 \leq \ell \leq L$ and every complex number \mbox{$s = t + i \beta \in \mathbb{C}$} we have
\[  \| \Gamma_{J_{L}}^{s}(A) - \Gamma_{J_{\ell}}^{s}(A) \| \, \leq \, 2 \, |J| \, e^{2|\beta| \Omega_{0} |J|} \, e^{8|s| \, \| \Phi\|_{\lambda}} \, \ell \, e^{(4|\beta| \Omega_{0} - \lambda) \, \lfloor \ell/2 \rfloor} \,.  \]
In particular, the time evolution operator $s \longmapsto \Gamma^{s}(A)$ is well-defined and quasilocal on the strip $\{ s \in \mathbb{C} \colon |\operatorname{Im}(s)| < \lambda/(4 \Omega_{0})\}$.
\end{Coro}

\begin{proof}
Let us assume that $\| A\| = 1$ and fix $j:= \lfloor \ell / 2 \rfloor$. Then, 
\begin{align*}
\| \Gamma_{J_{L}}^{s} & (A)  -  \Gamma_{J_{\ell}}^{s}(A) \|\\[3mm] 
& \leq \, \| \Gamma_{J_{L}}^{i \beta}(A) - \Gamma_{J_{j}}^{i \beta}(A)\|  + \| \Gamma_{J_{\ell}}^{i \beta}(A) - \Gamma_{J_{j}}^{i \beta}(A)\| + \| (\Gamma^{t}_{J_{L}} - \Gamma^{t}_{J_{\ell}}) \Gamma_{J_{j}}^{i \beta}(A) \|\,.
\end{align*}
The first two summands can be bounded using \eqref{equa:localityEstimatesExponentialDecay1}
as
\[ \| \Gamma_{J_{L}}^{i \beta}(A) - \Gamma_{J_{j}}^{i \beta}(A)\| + \| \Gamma_{J_{\ell}}^{i \beta}(A) - \Gamma_{J_{j}}^{i \beta}(A)\| \, \leq \, 2 \, e^{2|\beta| \Omega_{0} |J|} \, e^{4|\beta| \, \| \Phi\|_{\lambda}} \, e^{(4|\beta|\Omega_{0} - \lambda)j}\,. \]
The third one can be estimated using Theorem \ref{Theo:originalLiebRobinson}  and \eqref{equa:localityEstimatesExponentialDecay2},
\begin{align*}
\| (\Gamma^{t}_{J_{L}} - \Gamma^{t}_{J_{\ell}}) \Gamma_{J_{j}}^{i \beta}(A) \| \, & \, \leq \,  \| \Gamma_{J_{j}}^{i \beta}(A) \| \, |J_{j}| e^{4 |t| \Omega_{0}} \, e^{-\lambda(\ell - j)} e^{4|t| \, \| \Phi\|_{\lambda}}\\[2mm]
& \, \leq \, e^{2|\beta| \Omega_{0} |J|} \, e^{4 |\beta| \, \| \Phi\|_{\lambda}} \, |J_{j}| \, e^{4 |t| \Omega_{0}} \, e^{-\lambda(\ell - j)} e^{4|t| \, \| \Phi\|_{\lambda}}\\[2mm]
& \, \leq \, e^{6|s| \Omega_{0} |J|} \, e^{8|s| \, \| \Phi\|_{\lambda}} \, (\ell \, |J|) \, e^{- \lambda \, j}\,
\end{align*}
where in the last inequality we have used that $\ell-j \geq j$ by definition. Combining these estimates we conclude the result.
\end{proof}

\section{Local perturbations and expansionals}
\label{sec:expansionals}

In perturbation theory we study for a given Hamiltonian $H$ the effect of introducing a weak physical perturbance in the form of a new Hamiltonian $U$, namely $H + U$. It is then useful to relate $e^{- H}$ and $e^{-(H + U)}$ through identities of the form
\[ e^{ - (H + U)} \, = \,  E \, e^{ - H} \, E^{\dagger}  \]
for some suitable observable $E$ featuring the locality properties $U$.  Following Araki's approach  (see \cite[p. 135]{Araki69} or \cite{Araki73}) we take $E = E_{r}(U/2, H/2)$ where \\
\begin{equation}\label{equa:expansionals1} 
\begin{split}
 E_{r}(U;H)  & := e^{ - (H + U)} e^{H}\\  
 & = \mathbbm{1} + \sum_{m=1}^{\infty} (-1)^{m} \int_{0}^{1} d\beta_{1} \, \ldots \int_{0}^{\beta_{m}} d \beta_{m} \,\, \Gamma_{H}^{i \beta_{m}}(U) \ldots \Gamma_{H}^{i \beta_{1}}(U)\,,
 \end{split} 
 \end{equation}
and its inverse
\begin{equation}\label{equa:expansionals2} 
\begin{split}
E_{l}(U;H) &:= e^{-  H} e^{H + U} \\
& = \mathbbm{1} + \sum_{m=1}^{\infty}  \int_{0}^{1} d\beta_{1}  \ldots \int_{0}^{\beta_{m-1}} d \beta_{m} \,\, \Gamma_{H}^{i \beta_{1}}(U) \ldots \Gamma_{H}^{i \beta_{m}}(U)\,. 
\end{split}
\end{equation}
These identities are consequence of Duhamel's formula and allow us to apply the locality results on imaginary time evolution from the previous sections to obtain the following.

\begin{Theo}\label{Theo:BoundExpansionals}
Let us consider a quantum spin system on $\mathbb{Z}$ with bounded interaction $\Phi$. For every $\beta > 0$, $a \in \mathbb{Z}$  and $p \geq 0$ define
\[ E_{a,p}^{\beta}:= e^{-\beta H_{[a-p,a+1+p]}} \, e^{\beta H_{[a-p,a]} + \beta H_{[a+1,a+1+p]} }\,.  \]

%
%
%
%
%
%
%
%

\noindent Then, following the notation of Theorem \ref{Theo:FundamentalEstimations} , for every $0 \leq p \leq q$
\begin{align*}
\| E_{a,p}^{\beta}\|\, & \leq \, \exp{\left( \, \frac{1}{2}\, \sum_{k =1}^{p+1}  \, e^{4 \Omega_{0} \beta k} \, \Omega_{k}^{\ast}(4 \beta) \,  \right)}\,,\\[3mm]
 \| E_{a,p}^{\beta} - E_{a,q}^{\beta}\| & \leq  \, \frac{1}{2} \, \left( \sum_{k = p +1}^{q+1}  \, e^{4 \Omega_{0} \beta k} \, \Omega_{k}^{\ast}(4 \beta) \,  \right) \,  \exp{\left( \frac{1}{2} \, \sum_{k =1}^{q+1}  \, e^{4 \Omega_{0} \beta k} \, \Omega_{k}^{\ast}(4 \beta) \,  \right)}\,.
\end{align*}
Moreover, the same inequalities hold if we replace $\| E_{a,p}^{\beta}\|$ and $\| E_{a,p}^{\beta} - E_{a,q}^{\beta}\|$ with  $\| (E_{a,p}^{\beta})^{ -1}\|$ and $\| (E_{a,p}^{\beta})^{ -1} - (E_{a,q}^{\beta})^{ -1}\|$ repectively.
\end{Theo}

We can assume that $\beta = 1$ by absorbing $\beta$ in the interaction, since if we define $\Phi_{X}' = \beta \Phi_{X}$ for every finite $X \subset \mathbb{Z}$ the ineraction $\Phi'$ satisfies for every $k \geq 0$
\[ \Omega'_{k} = \beta \, \Omega_{k} \quad \mbox{ and } \quad \Omega'^{*}_{k}(x) = \Omega^*_{k}(\beta x )\,. \]
Denote $J:=[a,a+1]$. For each $p \geq 0$, let   $J_{p}:=[a-p,a+1+p]$ and 
\[ U_{J_{p}} := \sum\{ \Phi_{X}\colon X \subset J_{p}, X \cap [a-p,a] \neq \emptyset, X \cap [a+1,a+1+p] \neq \emptyset \}\,. \]
With the notation of \eqref{equa:expansionals1}  and \eqref{equa:expansionals2},
\begin{equation}\label{equa:expansionals3}
E_{a,p}^{1} = E_{l}(U_{J_{p}};H_{J_{p}})  \quad , \quad (E_{a,p}^{1})^{-1} =E_{r}(U_{J_{p}};H_{J_{p}})\,. 
\end{equation}
The strategy of the proof is then clear.

\begin{Lemm}\label{Lemm:applyHamiltonianToItself}
With the previous notation,  for every $0 \leq p \leq q$ 
\begin{align}
\label{equa:applyHamiltonianToItself1} \sup_{|s| \leq 1} \| \Gamma_{H_{J_{p}}}^{s}(U_{J_{p}})\| & \, \leq \, \frac{1}{2} \, \sum_{k = 1}^{p+1}  \, e^{4 \Omega_{0} k} \, \Omega_{k}^{\ast}(4)\,,\\[2mm]
\label{equa:applyHamiltonianToItself2}   \sup_{|s| \leq 1} \, \| \Gamma_{H_{J_{p}}}^{s}(U_{J_{p}}) - \Gamma_{H_{J_{q}}}^{s}(U_{J_{q}})\| & \, \leq \, \frac{1}{2} \, \sum_{k = p +1}^{q+1} \,  e^{4 \Omega_{0}  k} \, \Omega_{k}^{\ast}(4)\,. 
\end{align}
\end{Lemm}

\begin{proof}
Let us start with the following observation: for each $0< n_{1} \leq n_{2}$
\begin{align}\label{aux:applyHamiltonianToItself1} 
\sum_{n_{1} \leq j + k \leq n_{2}} \, e^{4 \Omega_{0}(j+k)} \, \Omega_{j} \, \Omega_{k}^{\ast}(4) \, \leq \, \frac{1}{4} \, \sum_{n_{1} \leq m \leq n_{2}}  \, e^{4  \Omega_{0} m} \, \Omega_{m}^{\ast}(4)\,. 
\end{align}
On the other hand, if we put $U_{-1}:=0$ then $U_{j}:= U_{J_{j}} - U_{J_{j-1}}$ has support in $J_{j}$ and satisfies the bound
\begin{equation}\label{aux:applyHamiltonianToItself2}  
\| U_{j} \| \, \leq \, 2 \, \Omega_{j+1}\,. 
\end{equation}
Having \eqref{aux:applyHamiltonianToItself1} and \eqref{aux:applyHamiltonianToItself2}  in mind, we get by Theorem \ref{Theo:FundamentalEstimations} that for $|s| \leq 1$
\begin{align*}
 \| \Gamma_{H_{J_{p}}}^{s}(U_{J_{p}})\|  \leq \sum_{j = 0}^{p}{\| \Gamma_{H_{J_{p}}}^{s}(U_{j})\|} \,&  \leq \, \sum_{j = 0}^{p} \| U_{j}\| \, e^{2  \Omega_{0} |J_{j}|} \, \left( \sum_{k = 0}^{p-j}e^{4\Omega_{0} k} \, \Omega_{k}^{\ast}(4) \right)\\[2mm]
 & \leq \sum_{j = 0}^{p} 2 \, \Omega_{j+1} \, e^{4 \Omega_{0} (j+1)} \, \left( \sum_{k = 0}^{p-j}e^{4 \Omega_{0} k} \, \Omega_{k}^{\ast}(4) \right) \\[2mm]
    & \leq  \frac{1}{2} \,  \sum_{m =1}^{p+1} \,  e^{4 \Omega_{0} m} \, \Omega_{m}^{\ast}(4)\,. 
\end{align*}
To prove the second bound, we can decompose
\begin{align*} 
\| \Gamma_{H_{J_{p}}}^{s}(U_{J_{p}}) - \Gamma_{H_{J_{q}}}^{s}(U_{J_{q}}) \| \,  \leq \,  \sum_{j=p+1}^{q} & \| \Gamma^{s}_{H_{J_{q}}}(U_{j})\| \, + \\
& + \, \sum_{j=0}^{p}{\| \Gamma_{H_{J_{p}}}^{s}(U_{j}) - \Gamma_{H_{J_{q}}}^{s}(U_{j})\|}\,. 
\end{align*}
We can now bound each term separately making use of Theorem \ref{Theo:FundamentalEstimations},
\begin{align*}
\sum_{j=0}^{p}{\| \Gamma_{H_{J_{p}}}^{s}(U_{j}) - \Gamma_{H_{J_{q}}}^{s}(U_{j})\|}\, & \, \leq \, \sum_{j=0}^{p} e^{2  \Omega_{0}|J_{j}|} \, \| U_{j}\| \,\left( \sum_{k = p-j+1}^{q-j} e^{4 \Omega_{0} k} \, \Omega_{k}^{\ast}(4)\right)\\[2mm]
\, & \, \leq \, \sum_{j=0}^{p} e^{4  \Omega_{0} (j+1)} \, 2 \, \Omega_{j+1} \,\left( \sum_{k = p-j+1}^{q-j} e^{4 \Omega_{0} k} \, \Omega_{k}^{\ast}(4)\right)\\[2mm]
\, &  \, \leq \, \frac{1}{2} \sum_{m = p+1}^{q+1}  \, e^{4 \Omega_{0} m} \, \Omega_{m}^{\ast}(4)\,
\end{align*}
and
\begin{align*}
\sum_{j=p+1}^{q}{\| \Gamma^{s}_{H_{J_{q}}}(U_{j})\|} & \, \leq \,  \sum_{j=p+1}^{q} e^{2 \Omega_{0} |J_{j}|} \, \| Q_{j}\| \, \left( \sum_{k = 0}^{q-j}{e^{4 \Omega_{0} k} \, \Omega_{k}^{\ast}(4)}\right)\\[2mm]
& \leq \sum_{j=p+1}^{q} e^{4 \Omega_{0} (j+1)} \, 2 \, \Omega_{j+1} \, \left( \sum_{k =0}^{q-j}{e^{4 \Omega_{0} k} \, \Omega_{k}^{\ast}(4)}\right)\\[2mm]
& \leq \frac{1}{2} \sum_{m = p+2}^{q+1}  \, e^{4 \Omega_{0} m} \, \Omega_{m}^{\ast}(4)\,. 
\end{align*}
Combining both inequalities we conclude the result.
\end{proof}

\begin{proof}[Proof of Theorem \ref{Theo:BoundExpansionals}]
We just explain the argument for $E_{a,p}^{1}$ since it is analogous for the inverse. Using  \eqref{equa:expansionals3} for $E_{a,p}^{1}$ and applying \eqref{equa:applyHamiltonianToItself1} to each factor of \eqref{equa:expansionals2}
\begin{align*} 
\| E_{a,p}^{1}\| \, & \leq \, \sum_{m=0}^{\infty}{ \left( \frac{1}{2} \, \sum_{k=1}^{p+1}  \, e^{4 \Omega_{0} k} \, \Omega^{\ast}_{k}(4 ) \right)^{m} } \frac{1}{m!}\, = \, \exp{\left( \frac{1}{2} \, \sum_{k =1}^{p+1}{  \, e^{4 \Omega_{0} k} \, \Omega^{\ast}_{k}(4)} \right)}\,.  
\end{align*}
On the other hand, notice that for every $0 \leq \beta_{1}, \ldots , \beta_{m} \leq 1$ we have applying \eqref{equa:applyHamiltonianToItself1} and \eqref{equa:applyHamiltonianToItself2} that
\begin{align*}
\Big\| \prod_{j=1}^{m} & { \Gamma_{H_{J_{p}}}^{\beta_{j}}(U_{J_{p}})}  -  \prod_{j=1}^{m}{\Gamma_{H_{J_{q}}}^{\beta_{j}}}(U_{J_{q}}) \Big\| \\[2mm]
& \, \leq \, \sum_{k=1}^{m} \, \prod_{j = 1}^{k-1}{\| \Gamma_{H_{J_{q}}}^{\beta_{j}}(U_{J_{q}})\|} \, \cdot \, \| \Gamma_{H_{J_{p}}}^{\beta_{j}}(U_{J_{p}}) - \Gamma_{H_{J_{q}}}^{\beta_{j}}(U_{J_{q}}) \| \cdot \prod_{j = k+1}^{m}{\| \Gamma_{H_{J_{p}}}^{\beta_{j}}(U_{J_{p}})\|}\\[2mm]
& \leq m \, \left( \frac{1}{2} \, \sum_{k = p+ 1}^{q+1}  \, e^{4 \Omega_{0} k} \, \Omega_{k}^{\ast}(4) \right) \,\left( \frac{1}{2} \, \sum_{k = 1}^{q+1}  \, e^{4 \Omega_{0} k} \, \Omega_{k}^{\ast}(4) \, \right)^{m-1}\,.
\end{align*}
Hence
\begin{align*}
\| E_{a,p}^{1} - E_{a,q}^{1}\| & \leq \sum_{m \geq 1} \, \frac{1}{m!} \, \sup_{0 \leq \beta_{1}, \ldots, \beta_{m} \leq 1} \Big\| \prod_{j=1}^{m}  { \Gamma_{H_{J_{p}}}^{\beta_{j}}(U_{J_{p}})}  -  \prod_{j=1}^{m}{\Gamma_{H_{J_{q}}}^{\beta_{j}}}(U_{J_{q}}) \Big\|\\[2mm]
& \leq  \, \frac{1}{2} \, \left( \sum_{k = p+1}^{q+1}  \, e^{4 \Omega_{0} k} \,  \Omega_{k}^{\ast}(4) \right)  \, \exp{\left(  \frac{1}{2} \sum_{k = 1}^{q+1}  \, e^{4 \Omega_{0} k} \, \Omega_{k}^{\ast}(4) \,  \right)}\,.
\end{align*}
\end{proof}

\section{Phase transitions and thermal states}
\label{sec:PhaseTransition}

Broadly speaking, phase transitions are singularities in correlation and thermodynamic functions of the temperature or the interactions in the thermodynamic limit. It is folklore that there is no thermal phase transition in one-dimension. To be more formal, let us restrict to interactions $\Phi$ on the one-dimensional lattice $\mathbb{Z}$. For each interval $[a,b] \subset \mathbb{Z}$  the \emph{local Gibbs state on $[a,b]$ at inverse temperature $\beta$} is defined as
\[ \varphi_{\beta}^{[a,b]}(Q) = \frac{\operatorname{tr}(e^{- \beta H_{[a,b]}} Q)}{\operatorname{tr}(e^{- \beta H_{[a,b]}})}\,, \quad Q \in \mathcal{A}_{\mathbb{Z}} \,. \]
For general C$^\ast$-dynamical systems, equilibrium states at inverse temperature $\beta$ are defined in terms of the KMS condition and may feature subtle issues such as nonexistence and nonuniqueness. Indeed, the existence and uniqueness of these KMS states is a feature of the abscence of phase transition. These properties are satisfied in our setting for rather general conditions on the interaction, see \cite{BraRob87}. In classical systems this is a result by Ruelle \cite{Ruelle1968}, while in quantum systems there are uniqueness results by Araki  \cite{Araki75} and Kishimoto \cite{Kishimoto76} (see also \cite[Section 6.2.5]{BraRob97}) that apply to interactions with uniformly bounded surface energies and, in particular, to bounded interactions $\Phi$ on $\mathbb{Z}$. They yield that for every $\beta > 0$ the following limit exists 
\[ \varphi_{\beta}(Q) = \lim_{\substack{a \rightarrow - \infty \\ b \rightarrow \infty}}{\varphi_{\beta}^{[a,b]}(Q)}\,, \quad Q \in \mathcal{A}\,.  \]
and defines the unique equilibrium state at inverse temperature $\beta$. 

Another feature of phase transitions is the clustering property of $\varphi_{\beta}$, namely the decay of the correlation function
\[ \operatorname{Corr}_{\beta}(Q_{1},Q_{2}) = \varphi_{\beta}(Q_{1},Q_{2}) - \varphi_{\beta}(Q_{1}) \varphi_{\beta}(Q_{2})\,. \] 
in terms of the distance that separates the support of the observables $Q_{1}$ and $Q_{2}$. That these correlations decay exponentially fast can be proven for rather general graphs at sufficiently high temperatures (small $\beta$) under suitable conditions on the interactions by means of cluster expansion techniques, see e.g. Kliesch et al. \cite{KGKRE14} in the finite range case, or the more recent result by Kuwahara et al. \cite{KuKaBra19} for the decay of the conditional mutual information even under long-range interactions. Similar ideas have been also applied to prove analiticity of thermodynamic and correlation functions at high enough temperatures, see \cite[Chapter III]{GrDo72}.

For low temperatures, it seems however that the typical arguments rely on the \emph{transfer operator method}. The idea is to construct for the given interaction $\Phi$ an operator $\mathcal{L}^{\Phi}$ on the space of observables. If $\Phi$ decays sufficiently fast, then $\mathcal{L}^{\Phi}$ has ``nice'' spectral properties in the spirit of Perron-Frobenius theorem, that in turn yield ``nice'' properties for our thermodynamic and correlation functions, see \cite{Baladi00, Ruelle2004} and for the non-commutative version \cite{Matsui01}.

In classical lattice systems, this approach was exploited by Ruelle \cite{Ruelle1968} to show that, under suitable polynomial decay assumptions on the interaction, correlations decay exponentially fast at every temperature. Moreover, Araki \cite{Araki69} showed that the pressure and correlation functions depend analitically on the temperature and the interaction if the latter decay exponentially fast, see also \cite{Ruelle2004}. This result was later improved by Dobrushin \cite{Dobrushin73} to more general interactions. By contrast, recall that there are examples of two-body interactions with polynomial decay in one-dimension featuring abscence of phase transition, see Dyson \cite{Dyson69}. We refer to \cite{GrDo72} and references therein for further information.

In quantum lattice systems, it seems that the only known result for low temperatures is due to Araki \cite{Araki69} for finite range interactions. Using the locality estimates from previous sections we extend Araki's results in the following form.

\begin{Theo}\label{Theo:noPhaseTranstionStrip}
Let us consider a quantum spin system over $\mathbb{Z}$ and a translational invariant interaction $\Phi$ with exponential decay $\| \Phi\|_{\lambda}<\infty$ for some $\lambda > 0$. Then, for each $\beta \in (0, \frac{\lambda}{2 \Omega_{0}})$ there exist constants $C,\delta > 0$ satisfying:\\[-2mm] 
\begin{enumerate}\itemsep0.7em
\item[(i)] For every $a,b \in \mathbb{Z}$ with $a < b$ and every $Q \in \mathcal{A}_{[a,b]}$  
\[ |\varphi_{\beta}(Q) - \varphi_{\beta}^{[a-k,b+k]}(Q)| \leq C e^{- \delta k}\,.\\ \]
\item[(ii)] For every $x \in \mathbb{Z}, k \in \mathbb{N}$, if $Q_{1} \in \mathcal{A}_{(- \infty, x]}$ and $Q_{2} \in \mathcal{A}_{[x+k, \infty)}$ then
\[ |\varphi_{\beta}(Q_{1}Q_{2}) - \varphi_{\beta}(Q_{1}) \, \varphi_{\beta}(Q_{2})| \, \leq \, C e^{ - \delta k}\,. \]
In other words, $\varphi_{\beta}$ has exponential decay of correlations.
\end{enumerate}
\end{Theo}
The proof follows the lines of the original argument by Araki \cite{Araki69}, its later extension by Golodets and Neshveyev \cite{GoNe98} to AF-algebras, and the more recent work on the non-commutative Ruelle transfer operator by Matsui \cite{Matsui01}.

\subsection{Idea of the proof. Notation}
 We can absorb $\beta$ in the interaction, and so we just have to prove the version $\beta = 1$ of the theorem assuming that $\lambda > 2 \Omega_{0}$. We analize the one-sided version of the problem on $\mathcal{A}_{\mathbb{N}}$. Let us consider for every $a \geq 1$
\[ \varphi^{[1,a]}\left(Q \right) = \frac{\operatorname{tr}\left(e^{-H_{[1,a]}}Q \right)}{\operatorname{tr}\left(e^{-H_{[1,a]}} \right)}\,, \quad Q \in \mathcal{A}_{\mathbb{N}}\,. \]
Fixed $1 \leq n < a$ let us denote
\begin{align*}
\widetilde{E}_{(n,a)}& :=e^{-\frac{1}{2}H_{[1,a]}} e^{\frac{1}{2} H_{[1+n,a]}}\\[2mm]
E_{(n,a)}& :=e^{-\frac{1}{2}H_{[1,a]}} e^{\frac{1}{2}H_{[1,n]} + \frac{1}{2} H_{[1+n,a]}}
\end{align*}
which satisfy
\begin{equation}\label{equa:ExpansionalDecompositionFormulas} 
e^{-\frac{1}{2}H_{[1,a]}} \, = \,  \widetilde{E}_{(n,a)} \, e^{-\frac{1}{2}H_{[1+n,a]}} \, = \, E_{(n,a)} \, e^{-\frac{1}{2}H_{[1,n]} - \frac{1}{2} H_{[1+n,a]}}\,.
\end{equation}

\begin{figure}[h]\label{Fig:ExpansionalDecompositionFormulas}
\begin{tikzpicture}[scale=0.89]

\pgfmathsetmacro{\step}{0.5}

\begin{scope}
\foreach \x in {1,...,10} \draw[fill=black] ({\step*\x},-0.2) circle [radius = 0.05];   

\draw (\step-0.1,0) rectangle ({\step*10 + 0.1},0.6);
\draw ({\step*5.5 + 0.1},0.3) node {\footnotesize $e^{- \frac{1}{2} \, H_{[1,a]}}$};
\end{scope}


\begin{scope}[xshift=6cm]
\foreach \x in {1,...,10} \draw[fill=black] ({\step*\x},-0.2) circle [radius = 0.05];   

\draw[fill=black!15] (\step-0.1,0) rectangle ({\step*10 + 0.1},0.6);
\draw[black] ({\step*5.5},0.3) node {\footnotesize $e^{- \frac{1}{2} \, H_{[1,a]}}$};

\draw[fill=black!15] (6*\step-0.1,0.8) rectangle ({10*\step + 0.1},1.4);
\draw[black] ({\step*8},1.1) node {\footnotesize $e^{ \frac{1}{2} \, H_{[1+n,a]}}$};

\draw (6*\step-0.1,1.6) rectangle ({10*\step + 0.1},2.2);
\draw[black] ({\step*8},1.9) node {\footnotesize $e^{- \frac{1}{2} \, H_{[1+n,a]}}$};
\end{scope}


\begin{scope}[xshift=12cm]
\foreach \x in {1,...,10} \draw[fill=black] ({\step*\x},-0.2) circle [radius = 0.05];   

\draw[fill=black!15] (\step-0.1,0) rectangle ({\step*10 + 0.1},0.6);
\draw[black] ({\step*5.5},0.3) node {\footnotesize $e^{- \frac{1}{2} \, H_{[1,a]}}$};

\draw[fill=black!15] (6*\step-0.1,0.8) rectangle ({10*\step + 0.1},1.4);
\draw[black] ({\step*8},1.1) node {\footnotesize $e^{ \frac{1}{2}  H_{[1+n,a]}}$};

\draw (6*\step-0.1,1.6) rectangle ({10*\step + 0.1},2.2);
\draw[black] ({\step*8},1.9) node {\footnotesize $e^{- \frac{1}{2}  H_{[1+n,a]}}$};

\draw[fill=black!20] (1*\step-0.1,0.8) rectangle ({5*\step + 0.1},1.4);
\draw[black] ({\step*3},1.1) node {\footnotesize $e^{ \frac{1}{2}  H_{[1,n]}}$};

\draw (1*\step-0.1,1.6) rectangle ({5*\step + 0.1},2.2);
\draw[black] ({\step*3},1.9) node {\footnotesize $e^{- \frac{1}{2}  H_{[1,n]}}$};

\end{scope}

\end{tikzpicture}
\caption{\small From left to right, decomposition of $e^{-\frac{1}{2}H_{[1,a]}}$ in terms of the factors $\widetilde{E}_{(n,a)}$ and $E_{(n,a)}$ represented both with shaded boxes.}
\end{figure}
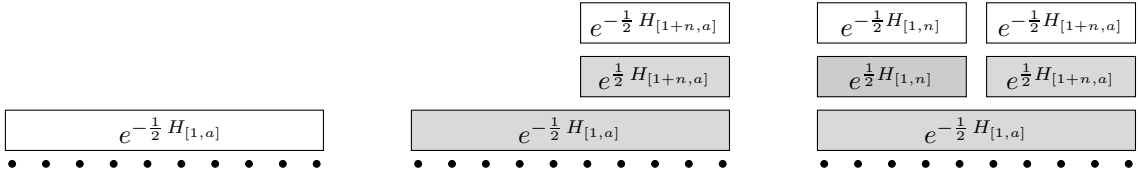

\noindent Using that $\tau_{-n}$ is an algebra homomorphism from $\mathcal{A}_{[1+n,\infty)}$ into $\mathcal{A}_{[1,\infty)}$, we have
\begin{align*}
\operatorname{tr}\left( e^{-H_{[1,a]}}Q\right) & \, = \, \operatorname{tr}\big( \,e^{-H_{[1+n,a]}} \, \widetilde{E}_{(n,a)}^{\dagger} \,Q \, \widetilde{E}_{(n,a)} \, \big)\\[2mm]
& \, = \, \operatorname{tr}\left(\,  e^{-H_{[1,a-n]}} \, \mathcal{L}_{(n,a)}(Q)\, \right)\,
\end{align*}
where $\mathcal{L}_{(n,a)}:\mathcal{A}_{\mathbb{N}} \longrightarrow \mathcal{A}_{\mathbb{N}} $ is given by
\[ \mathcal{L}_{(n,a)}(Q) := \tau_{-n} \operatorname{tr}_{[1,n]}\big(\widetilde{E}_{(n,a)}^{\dagger} Q \widetilde{E}_{(n,a)} \big)\, = \, \tau_{-n} \operatorname{tr}_{[1,n]}\big(e^{-H_{[1,n]}}\, E_{(n,a)}^{\dagger} Q E_{(n,a)} \big).\]
Therefore, we can rewrite local Gibbs state as
\begin{equation}\label{equa:localGibbsL}
\varphi^{[1,a]}(Q) \, = \, \frac{\varphi^{[1,a-n]}\big( \mathcal{L}_{(n,a)}(Q)\big)}{\varphi^{[1,a-n]}\big( \mathcal{L}_{(n,a)}(\mathbbm{1})\big)}\,.
\end{equation}

The rest of the proof has been decomposed into several stages corresponding to the following subsections. Let us briefly sketch them here:\\

\begin{enumerate}
\item[(1)] The estimates on the expansionals given in subsection \eqref{sec:expansionals} yield that for every $n \in \mathbb{N}$ the elements $E_{(n,a)}$  converge as $a \rightarrow + \infty$ to a certain $E_{n}$ exponentially fast. In particular, the family $\mathcal{L}_{(n,a)}$ strongly converges to some operator $\mathcal{L}_{n}$.\\
\item[(2)] The sequence $\mathcal{L}_{n}$ satisfies  $\mathcal{L}_{n} = \mathcal{L} \circ \ldots \circ \mathcal{L}$ ($n$ times) where  $\mathcal{L} : = \mathcal{L}_{1}$ is the so-called \emph{Ruelle transfer operator} that maps $\mathcal{A}_{\mathbb{N}}(x)$ into $\mathcal{A}_{\mathbb{N}}(x)$ for a suitable constant $x>1$ (recall Section \ref{sec:notation} for nomenclature). Moreover, there exist a positive constant $\mu >0$, a state $\nu$ over $\mathcal{A}_{\mathbb{N}}$ and a positive and invertible observable $h$ in $\mathcal{A}_{\mathbb{N}}(x)$ such that the rank-one operator
\[ \mathcal{P}:\mathcal{A}_{\mathbb{N}}(x) \longrightarrow \mathcal{A}_{\mathbb{N}}(x) \,, \quad Q \,\, \longmapsto \,\,  \nu(Q) h \]
satisfies\\[-3mm]
\begin{enumerate}\itemsep0.5em 
\item[(i)]  $\mathcal{P} \circ \mathcal{P} = \mathcal{P}$,
\item[(ii)]  $\mathcal{L} \circ \mathcal{P} = \mathcal{P} \circ \mathcal{L} = \mu \, \mathcal{P}$.\\[-3mm]
\end{enumerate}

\noindent In particular, we can decompose for every $n \in \mathbb{N}$
\begin{align*} 
\mathcal{L}^{n} = \mu^{n} \, \mathcal{P} \, + \, \mathcal{L}^{n}(\mathbbm{1} - \mathcal{P})\,.
\end{align*}
\item[(3)] The rescaled operator $L := \mu^{-1} \mathcal{L}: \mathcal{A}_{\mathbb{N}}(x) \longrightarrow \mathcal{A}_{\mathbb{N}}(x)$ satisfies that there is $\gamma >1$ such that
\[  \lim_{n \rightarrow \infty}{\gamma^{n}} \, \|| L^{n}(\mathbbm{1} - \mathcal{P}) |\|_{1,x} = 0\,.\] 
\item[(4)] The sequence $\nu \circ \tau_{n}$ converges to a translation invariant state $\psi^{[1,\infty)}$ on $\mathcal{A}_{[1,\infty)}$, which can be extended to a (unique) translation invariant state $\psi$ over $\mathcal{A}_{\mathbb{Z}}$ satisfying exponential decay of correlations. The state $\psi$ is actually the infinite-volume KMS state associated to $\Phi$ at $\beta=1$.
\end{enumerate}


\subsection{The map $\mathcal{L}_{n}$. Definition and basic properties}

We have to prove for a fixed $n \in \mathbb{N}$ that $E_{(n,a)}$, $\widetilde{E}_{(n,a)}$ and their adjoint converge as $a$ tends to infinity. For that we will make use of Theorem \ref{Theo:BoundExpansionals}. Let us consider the slightly modified version of the error terms that appear there
\begin{equation}\label{equa:constantsExpansionalsConv} 
\mathcal{G} := \exp{\left( \, \mbox{$\sum_{k \geq 1}$}  \, e^{2 \Omega_{0} k} \, \Omega_{k}^{\ast}(2) \, \right)} \quad , \quad \mathcal{G}_{\ell}:= \mathcal{G} \,\, \mbox{$\sum_{k \geq \ell}$} \, e^{2 \Omega_{0} k} \, \Omega_{k}^{\ast}(2)\quad \ell \geq 1\,. 
\end{equation}
Using the exponential decay condition $\| \Phi\|_{\lambda} < \infty$ we can argue as in the proof of \eqref{equa:localityEstimatesExponentialDecay1}-\eqref{equa:localityEstimatesExponentialDecay2} to deduce that
\[ \mathcal{G} \leq \exp{\big( e^{2 \, \| \Phi\|_{\lambda}}  \big)} \]
and
\begin{align*} 
\sum_{\ell \geq 1}{\mathcal{G}_{\ell} \, e^{(\lambda - 2 \Omega_{0}) \ell}} & \, \leq \, \exp{\big( e^{2 \, \| \Phi\|_{\lambda}}  \big)}\, \sum_{ \ell \geq 1} \, \sum_{k \geq \ell} e^{(2\Omega_{0} - \lambda) \, (k-\ell)} e^{\lambda \, k} \, \Omega^{\ast}_{k}(2)\\[2mm] 
& \, = \, \exp{\big( e^{2 \, \| \Phi\|_{\lambda} } \big)} \, \sum_{k\geq 1} \, \sum_{1 \leq \ell \leq k} \, e^{(2\Omega_{0} - \lambda) \, (k-\ell)} e^{\lambda \, k} \, \Omega^{\ast}_{k}(2) \\[2mm]
& \, \leq \, \exp{\big( e^{2 \, \| \Phi\|_{\lambda} } \big)} \, \big( \,  \mbox{$\sum_{\ell \geq 0}$} \, e^{(2 \Omega_{0} - \lambda)\ell} \, \big) \, \big(\,  \mbox{$\sum_{k \geq 1}$} \, e^{\lambda k} \, \Omega_{k}^{\ast}(2) \, \big)\\[4mm]
& \, \leq  \, \exp{\big(2\| \Phi\|_{\lambda} + e^{2 \, \| \Phi\|_{\lambda}}\big)} \, \big( \, \mbox{$\sum_{\ell \geq 0}$} \, e^{(2 \Omega_{0} - \lambda)\ell}\big) \, < \, \infty\,.
\end{align*}
In particular, the latter yields that $\mathcal{G}_{\ell}$ decays exponentially fast. We have then by Theorem \ref{Theo:BoundExpansionals} the following convergence result.

\begin{Prop}\label{Prop:expansionalProperties}
For each $n \in \mathbb{N}$ the following limits exist and are invertible
\[ E_{n} := \lim_{a \rightarrow + \infty}E_{(n,a)}\,, \quad \widetilde{E}_{n} := \lim_{a \rightarrow + \infty}\widetilde{E}_{(n,a)} = E_{n} \, e^{-\frac{1}{2}H_{[1,n]}}\,. \]
Moreover, they satisfy for every $1 \leq n < a$
\begin{enumerate}\itemsep0.7em
\item[(i)] $\|E_{n}\|, \| E_{(n,a)}\|, \| E_{n}^{-1}\|, \| E_{(n,a)}^{-1}\| \leq \mathcal{G}$,
\item[(ii)] $\| E_{(n,a)} - E_{n} \|, \| E_{(n,a)}^{-1} - E_{n}^{-1} \| \leq \mathcal{G}_{a-n}$,
\item[(iii)] For every positive and invertible $Q \in \mathcal{A}_{\mathbb{N}}$
\[ \mathcal{G}^{-2} \, \| Q^{-1}\|^{-1} \,\mathbbm{1} \, \leq \, E^{\dagger} Q E \, \leq \, \mathcal{G}^{2} \, \| Q\| \, \mathbbm{1}\,, \quad E \in \{ E_{(n,a)}, E_{n} \}\,. \]
\end{enumerate}
\end{Prop}

\begin{proof}
Item (iii) is straightforward from (i) and (ii) so let us focus on the latter. We just explain those inequalities involving $E_{(n,a)}$ and $E_{n}$, the others are analogous. Fixed $n \in \mathbb{N}$ and $a,c \in \mathbb{N}$ satisfying $n \leq c \leq a$ and $c \leq 2n$, it is easy to see applying Theorem \ref{Theo:BoundExpansionals} on a suitable Hamiltonian that
\begin{equation}\label{equa:boundAuxExpansionalsConvergence} 
\| E_{(n,a)} - e^{-\frac{1}{2}H_{[n-c,c]}}e^{\frac{1}{2}H_{[2n+1-c,n]} + \frac{1}{2}H_{[n+1,c]}}\|  \leq \frac{1}{2} \mathcal{G}_{c-n}\,. 
 \end{equation}
 This yields $\| E_{(n,a)} - E_{(n,a')}\| \leq \mathcal{G}_{a-n}$ for every $a,a' \in \mathbb{N}$ with $n \leq a \leq a'$ and so $E_{(n,a)}$ converges to some $E_{n}$ as $a$ tends to infinity. Taking limit when $a'$ tends to infinity in the previous expression yields (ii). The bound $\| E_{(n,a)}\| \leq \mathcal{G}$ valid for every $1 \leq n< a$ by Theorem \ref{Theo:BoundExpansionals} extends therefore to $E_{n}$, and so (i) holds. 
\end{proof}

\begin{Defi}
For every $n \in \mathbb{N}$ let $\mathcal{L}_{n}:\mathcal{A}_{\mathbb{N}} \longrightarrow \mathcal{A}_{\mathbb{N}}$ be the positive linear operator 
\[ \mathcal{L}_{n}(Q) :=  \tau_{-n} \, \operatorname{tr}_{[1,n]}\Big(  e^{- H_{[1, n]} } \,    E_{n}^{\dagger }\,  Q \, E_{n} \Big) \, = \, \tau_{-n} \, \operatorname{tr}_{[1,n]}\Big(       \widetilde{E}_{n}^{\dagger }\,  Q \, \widetilde{E}_{n} \Big) \,. \]
We will simply denote $\mathcal{L}:= \mathcal{L}_{1}$.
\end{Defi}

\begin{Theo}\label{Theo:PropertiesLn}
Let $n, \ell \in \mathbb{N}$ and $Q,A \in \mathcal{A}_{\mathbb{N}}$. Then,\\[-2mm]
\begin{enumerate}\itemsep0.7em
\item[(i)] $\| \mathcal{L}_{n}(Q)\| \leq \operatorname{tr}(e^{-H_{[1,n]}}) \, \mathcal{G}^{2} \, \| Q\|$,
\item[(ii)] $\| \mathcal{L}_{n}(Q) - \mathcal{L}_{(n,n+\ell)}(A)\| \, \leq \, \operatorname{tr}(e^{-H_{[1,n]}}) \, \big(\, 2 \, \mathcal{G} \, \mathcal{G}_{\ell} \, \| Q\| + \mathcal{G}^{2} \, \| Q - A\| \big)$,
\item[(iii)] $\| \mathcal{L}_{n}(Q)\|_{\ell} \leq \operatorname{tr}(e^{-H_{[1,n]}}) \, (2 \mathcal{G} \mathcal{G}_{\ell} + \mathcal{G}^{2}\| Q\|_{n+\ell})$. \\
\end{enumerate}
Moreover, if $Q \in \mathcal{A}_{\mathbb{N}}$ is positive and invertible, then\\[-2mm]
\begin{enumerate}\itemsep0.7em
\item[(iv)] $\operatorname{tr}(e^{-H_{[1,n]}}) \, \mathcal{G}^{-2} \, \| Q^{-1}\|^{-1} \, \mathbbm{1}\, \leq \, \mathcal{L}_{n}(Q) \, \leq  \,\operatorname{tr}(e^{-H_{[1,n]}}) \, \mathcal{G}^{2} \, \| Q\| \, \mathbbm{1}$\,,
\item[(v)] $ \| \mathcal{L}_{n}(Q)\|_{\ell} \, \| \mathcal{L}_{n}(Q)^{-1}\|  \, \leq \, 2 \, \mathcal{G}^{3} \, \mathcal{G}_{\ell} \, \| Q\| \, \| Q^{-1}\| + \mathcal{G}^{4} \, \| Q\|_{n+\ell} \, \| Q^{-1} \|$.
\end{enumerate}
\end{Theo}

\begin{proof}
All statements are consequence of Proposition \ref{Prop:expansionalProperties}.
Items (i) and (iv) are straightforward. Let us check (ii).
\begin{equation*}\label{equa:FlatAux1}
\| \, \mathcal{L}_{n}(Q) - \mathcal{L}_{(n,n+\ell)}(A) \, \| \,   \leq \, \operatorname{tr}\big( e^{-H_{[1, n]}}\big) \, \| E_{n}^{\dagger} Q E_{n} - E_{(n,n+\ell)}^{\dagger} A E_{(n,n+\ell)}^{\dagger}\|
\end{equation*}
where 
\begin{equation*}\label{equa:FlatAux2}
\begin{split}
\| E_{n}^{\dagger} Q E_{n} - E_{(n,n+\ell)}^{\dagger} A E_{(n,n+\ell)}  \| & \leq  \| E_{n}^{\dagger} Q E_{n} - E_{(n,n+\ell)}^{\dagger} Q E_{(n,n+\ell)} \| + \mathcal{G}^{2} \, \| Q - A\|\\[2mm]
& \leq 2 \, \mathcal{G}  \, \mathcal{G}_{\ell} \, \| Q\| + \mathcal{G}^{2} \| Q - A \| \,.
\end{split}
\end{equation*}
To prove (iii) for the given $Q$, let us take  $\widetilde{Q} \in \mathcal{A}_{[1,n+\ell]}$ satisfying $\| Q - \widetilde{Q}\| = \| Q\|_{n+\ell}$ which exists by compactness. Then, $\mathcal{L}_{(n, n+\ell)}(\widetilde{Q})$ belongs to $\mathcal{A}_{[1, \ell]}$ and so by (ii)
\[ \| \mathcal{L}_{n}(Q)\|_{\ell} \, \leq \, \| \mathcal{L}_{n}(Q) - \mathcal{L}_{(n,n+\ell)}(\widetilde{Q}) \| \, \leq \, \operatorname{tr}(e^{-H_{[1,n]}}) \, \big(\, 2 \, \mathcal{G} \, \mathcal{G}_{\ell} \, \| Q\| + \mathcal{G}^{2} \, \| Q\|_{n+\ell} \big)\,. \]
Finally, combining (iii) and (iv) we conclude that (v) holds.
\end{proof}

A consequence of the pervious theorem is that for every positive and invertible element $Q \in \mathcal{A}_{\mathbb{N}}$ it holds $\| \mathcal{L}_{n}(Q)\| \, \|\mathcal{L}_{n}(Q)^{-1} \| \leq \mathcal{G}^{4} \| Q\|\, \| Q^{-1}\|$. The next result improves this estimates for large values of $n$.

\begin{Theo}\label{Theo:Flatness}
Let $Q \in \mathcal{A}_{\mathbb{N}}$ be positive and invertible and let $n, \ell \in \mathbb{N}$ such that $1 + \ell \leq n$. Then 
\begin{align*}
\| \mathcal{L}_{n}(Q)\|  \| \mathcal{L}_{n}(Q)^{-1}\| \, & \leq \, \mathcal{G}^{4}   \left( 1 + 2 \, \mathcal{G}^{3} \, \mathcal{G}_{\ell} \, \| Q\| \, \| Q^{-1}\| \,  + 2 \, \mathcal{G}^{4} \, \| Q\|_{n-\ell }\, \| Q^{-1}\| \right)\,. 
\end{align*}
\end{Theo}

\begin{proof}
Let us write
\[ E_{(n,n+\ell)}' \, := \, e^{-\frac{1}{2}H_{[n-\ell+1, n+\ell]}} \, e^{\frac{1}{2}H_{[n-\ell+1, n]}} \, e^{\frac{1}{2} H_{[1+n,n+\ell]}},, \]
whose support  is contained in $ [ n - \ell +1, n + \ell ]$. Note that \eqref{equa:boundAuxExpansionalsConvergence} yields
\begin{equation}\label{equa:boundAuxExpansionalsConvergence2}
\| E_{(n,n+\ell)}' - E_{(n,n+\ell)} \| \leq \mathcal{G}_{\ell}\,.
\end{equation}
Using the tracial state over $\mathcal{A}_{[n-\ell+1, \infty)}$ we can define
\[ A := \operatorname{tr}_{[ n-\ell+1, \infty)} (Q) \in  \mathcal{A}_{[1,n-\ell]} \]
which is positive and invertible. Notice that $A$ and $E_{(n,n+\ell)}'$ have disjoint support and thus
\begin{align*}
\widetilde{A} \, & := \tau_{-n} \, \operatorname{tr}_{[1,n]}\big( e^{-H_{[1,n]} } E_{(n,n+\ell)}'^{\, \dagger} \, A \,  E_{(n,n+\ell)}' \big)\\[2mm]
& =  \tau_{-n} \, \operatorname{tr}_{[1,n]}\big( e^{-H_{[1,n]} } A^{1/2} E_{(n,n+\ell)}'^{\, \dagger} E_{(n,n+\ell)}'  A^{1/2} \big)  \,.
\end{align*}
Consequently, 
\begin{equation}\label{equa:FlatnessAux1}
\mathcal{G}^{-2} \, \operatorname{tr}\big( e^{- H_{[1,n]}} A \big) \mathbbm{1}  \, \leq \, \widetilde{A} \,\leq \, \mathcal{G}^{2} \, \operatorname{tr}\big( e^{- H_{[1,n]}} A \big) \mathbbm{1}\,. 
\end{equation}
Let $\phi$ and $\phi'$ be two states over $\mathcal{A}_{\mathbb{N}}$. From \eqref{equa:FlatnessAux1} it follows that
\[ \phi(\widetilde{A}) \, \leq \, \mathcal{G}^{4} \, \phi'(\widetilde{A})\, \]
and thus
\begin{align*}  
\phi(\mathcal{L}_{n}(Q)) \, & \leq \, \phi(\widetilde{A}) + \| \, \mathcal{L}_{n}(Q) - \widetilde{A} \, \|\\[2mm]  
& \, \leq \, \mathcal{G}^{4}  \, \phi'(\widetilde{A}) + \| \, \mathcal{L}_{n}(Q)- \widetilde{A} \, \|\\[2mm]  
& \, \leq \, \mathcal{G}^{4} \, \phi'(\mathcal{L}_{n}(Q)) + (1 + \mathcal{G}^{4}) \, \| \, \mathcal{L}_{n}(Q) - \widetilde{A} \, \|\,.
\end{align*}
Using that for every positive and invertible $B \in \mathcal{A}_{\mathbb{N}}$
\[ \| B\| = \sup_{\phi}{\phi(B)} \quad \quad \mbox{ and } \quad \quad \| B^{-1}\|^{-1} = \inf_{\phi'}{\phi'(B)}\,, \]
where the supremum and infimum are both taken with respect to all states over $\mathcal{A}_{\mathbb{N}}$, we conclude that
\begin{equation}\label{equa:FlatnessAux2} 
\begin{split}
\| \mathcal{L}_{n}(Q) \| \, \| \mathcal{L}_{n}(Q)^{-1}\| \, & \leq  \, \mathcal{G}^{4} + (1 + \mathcal{G}^{4}) \, \| \mathcal{L}_{n}(Q) - \widetilde{A} \| \, \| \mathcal{L}_{n}(Q)^{-1}\|\\[2mm]
& \leq \mathcal{G}^{4} \, \left( 1 + 2 \, \| \mathcal{L}_{n}(Q) - \widetilde{A} \| \, \| \mathcal{L}_{n}(Q)^{-1}\| \right)\,.
\end{split}
\end{equation}
But notice that the same argument given in the proof of Theorem \ref{Theo:PropertiesLn}.(v) using \eqref{equa:boundAuxExpansionalsConvergence2} leads to
\begin{equation} \label{equa:FlatnessAux3}
\| \mathcal{L}_{n}(Q) - \widetilde{A}\| \, \| \mathcal{L}_{n}(Q)^{-1}\| \, \leq \, 2 \, \mathcal{G}^{3} \, \mathcal{G}_{\ell} \, \| Q\| \, \| Q^{-1}\| \,  + \mathcal{G}^{4} \, \| Q - A\|\, \| Q^{-1}\| \,.   
\end{equation}
To bound the last summand, we use that taking \mbox{$Q_{n- \ell} \in \mathcal{A}_{[1,n-\ell]}$} with 
\[ \| Q - Q_{n -\ell}\|=\| Q\|_{n- \ell} \] 
it holds that
\begin{equation}\label{equa:FlatnessAux4}
\| Q - A\| = \| Q- Q_{n -\ell} + \operatorname{tr}_{[n - \ell ,  \infty)} (Q_{n -\ell} - Q) \| \leq  2 \,  \| Q\|_{n- \ell}\,. 
\end{equation}
Combining \eqref{equa:FlatnessAux2}, \eqref{equa:FlatnessAux3}  and \eqref{equa:FlatnessAux4}, we conclude the result.  
\end{proof}

\subsection{The maps $\mathcal{L}$ and $L$: fixed points}

\begin{Prop}\label{Prop: LisIterative}
For every $n \in \mathbb{N}$,  $\mathcal{L}_{n} = \mathcal{L}^{n} := \mathcal{L} \circ \ldots \circ \mathcal{L} \,\,\, (n \mbox{ times})$.
\end{Prop}

\begin{proof}
For every $1\leq n < a$, it is easy to check that
\[ \widetilde{E}_{(n,a+n)} \, \tau_{n}\big( \widetilde{E}_{(1,a)} \big) = \widetilde{E}_{(n+1, n+a)}\,, \]
and so
\[ \mathcal{L}_{(1,a)} \circ \mathcal{L}_{(n,a+n)} = \mathcal{L}_{(n+1, a+n)}\,. \]
By  Theorem \ref{Theo:PropertiesLn}, we can take limit in the previous expression when $a \rightarrow + \infty$ to get that $\mathcal{L}_{1} \circ \mathcal{L}_{n} = \mathcal{L}_{n+1}$.
\end{proof}

\begin{Theo}\label{Theo:fixedState}
There exist a state $\nu$ over $\mathcal{A}_{\mathbb{N}}$ and a real number $\mu > 0$ such that
\begin{equation}\label{equa:fixedStateAux0} 
\nu(\mathcal{L}(Q)) = \mu \, \nu(Q) \quad \mbox{ for every } Q \in \mathcal{A}_{\mathbb{N}}\,. 
\end{equation}
Moreover, it satisfies that for every $n \in \mathbb{N}$
\begin{equation}\label{equa:fixedStateAux}
\mathcal{G}^{-2} \, \operatorname{tr}\big( e^{- H_{[1,n]}} \big) \,  \leq \,  \mu^{n} \, \leq \, \mathcal{G}^{2} \, \operatorname{tr}\big( e^{- H_{[1,n]}} \big)\,. 
\end{equation}
\end{Theo}

\begin{proof}
Let $\mathcal{L}: \mathcal{A}_{\mathbb{N}} \longrightarrow \mathcal{A}_{\mathbb{N}} $ be a positive linear operator on a C*-algebra $\mathcal{A}_{\mathbb{N}} $ satisfying that $\mathcal{L}(\mathbbm{1}) \geq \gamma \mathbbm{1}$ for a positive constant $\gamma > 0$.  The dual operator $\mathcal{L}^{\ast}:\mathcal{A}_{\mathbb{N}}  \longrightarrow \mathcal{A}_{\mathbb{N}} $ is weak$^\ast$-weak$^\ast$ continuous, and so does its restriction $\mathcal{L}^{\ast} : \mathcal{S} \longrightarrow \mathcal{A}_{\mathbb{N}} ^{\ast}$ to  the weak$^\ast$-compact convex set $\mathcal{S}$ of states over a  $\mathcal{A}_{\mathbb{N}} $. Next, let us define
\[ \widehat{\mathcal{L}} : \mathcal{S} \longrightarrow \mathcal{S} \,, \quad \widehat{\mathcal{L}} (\varphi) = \frac{\mathcal{L}^{\ast} \varphi }{\mathcal{L}^{\ast} \varphi (\mathbbm{1})}\,. \]
Notice that it is well-defined, as for every $\varphi \in \mathcal{S}$ the map $\widehat{\mathcal{L}}(\varphi)$ is linear and positive since
 \begin{equation}\label{equa:the_map_aux1} 
 \mathcal{L}^{\ast} \varphi(\mathbbm{1}) = \varphi(\mathcal{L} ( \mathbbm{1})) \geq \gamma \varphi(\mathbbm{1}) = \gamma > 0\,, 
 \end{equation}
 and moreover $\widehat{\mathcal{L}}(\varphi)(\mathbbm{1}) = 1$. It is also weak$^\ast$-weak$^\ast$-continuous again by \eqref{equa:the_map_aux1}. We are in the conditions to apply Schauder-Tychonov theorem to deduce the existence of a fixed point $\nu \in \mathcal{S}$ of $\widehat{\mathcal{L}}$, which clearly satisfies for every $Q \in \mathcal{A}_{\mathbb{N}} $
 \[ \nu(\mathcal{L}(Q)) = \mu \, \nu(Q) \quad \mbox{ where } \quad  \mu := \nu(\mathcal{L}(\mathbbm{1}))\,. \] 
To check second statement, we simply apply the state $\nu$ to each term in Theorem \ref{Theo:PropertiesLn}.(iv) with $Q = \mathbbm{1}$.
\end{proof}

\begin{Defi}
Let us denote by $L: \mathcal{A}_{\mathbb{N}} \longrightarrow \mathcal{A}_{\mathbb{N}} $ the operator given by $L = \mathcal{L}/\mu$, where $\mu > 0$ is the positive number provided by Theorem \ref{Theo:fixedState}. We will also write $L_{(n,a)} = \mathcal{L}_{(n,a)}/\mu^{n}$ for every $a>n \geq 1$.
\end{Defi}

The next result collects some straightforward properties of $L$ inherited from $\mathcal{L}$.

\begin{Coro}\label{Coro:Lbounded}
Let $n, \ell \in \mathbb{N}$ and $Q \in \mathcal{A}_{\mathbb{N}}$.Then\\[-2mm]
\begin{enumerate}\itemsep0.7em
\item[(i)] $\| L^{n}(Q) \| \leq \mathcal{G}^{4} \, \| Q\|$\,,
\item[(ii)] $\|L^{n}(Q) - L_{(n,n+\ell)}(Q)\| \leq 2 \mathcal{G}^{3} \mathcal{G}_{\ell} \| Q \|$\,,
\item[(iii)] $\| L^{n}(Q)\|_{\ell} \leq 2 \mathcal{G}^{3} \mathcal{G}_{\ell} \, \| Q\| + \mathcal{G}^{4} \| Q\|_{n + \ell}$\,.\\[-2mm]
\end{enumerate}
Moreover, if $Q$ is positive and invertible\\[-2mm]
\begin{enumerate}
\item[(iv)] $\mathcal{G}^{-4} \, \| Q^{-1}\|^{-1} \, \mathbbm{1} \, \leq \, L^{n}(Q) \, \leq \, \mathcal{G}^{4} \, \| Q\| \, \mathbbm{1}$.\\[-2mm]
\end{enumerate}
In particular, for every $x>1$ with $\sum_{\ell \geq 1}{\mathcal{G}_{\ell}\, x^{\ell}} < \infty$ we have that the restriction $L: \mathcal{A}_{\mathbb{N}} (x) \longrightarrow \mathcal{A}_{\mathbb{N}} (x)$ is well-defined and continuous. Indeed, there is a constant $C_{x} > 0$ such that for every $m,n \in \mathbb{N}$ and every $Q \in \mathcal{A}_{\mathbb{N}} (x)$
\[ \|| L^{n}(Q) |\|_{m,x} \leq C_{x} \, \|| Q |\|_{n+m,x}\,. \]
\end{Coro}

\begin{proof}
Items (i)-(iv) follow from Theorem \ref{Theo:PropertiesLn}.(i)-(iv), dividing by $\mu^{n}$ on both sides of each inequality and using \eqref{equa:fixedStateAux}. To check the last statement, let $Q \in \mathcal{A}_{\mathbb{N}}(x)$ with $x > 1$ satisfying $\sum_{\ell \geq 1}{\mathcal{G}_{\ell} \, x^{\ell}}< \infty$ and let $m,n \in \mathbb{N}$. Using (i), (ii) and (iv) we can estimate
\begin{align*}
\|| L^{n}(Q)|\|_{m,x} & = \| L^{n}(Q)\| + \sum_{\ell \geq m}{\| L^{n}(Q)\|_{\ell} \, x^{\ell}}\\[2mm]
& \leq \| L^{n}(Q)\| + 2\,  \mathcal{G}^{3} \, \| Q\| \, \sum_{\ell \geq m}{\mathcal{G}_{\ell} x^{\ell}} + \mathcal{G}^{4} \, \sum_{\ell \geq m}{\| Q\|_{n+\ell} \, x^{\ell}}\\[2mm]
& \leq \Big(  \mathcal{G}^{4} + 2 \mathcal{G}^{3} \, \sum_{\ell \geq m}{\mathcal{G}_{\ell} x^{\ell}} + \mathcal{G}^{4} x^{-n} \Big) \, \| |Q| \|_{m+n,x} \,.
\end{align*}
\end{proof}

\begin{Theo}\label{Theo:fixedPointH} 
The map $L:\mathcal{A}_{\mathbb{N}} \longrightarrow \mathcal{A}_{\mathbb{N}}$ has a fixed point $h \in \mathcal{A}_{\mathbb{N}}$ satisfying:
\begin{enumerate}
\item[(i)] $h \in \mathcal{A}_{\mathbb{N}}(x)$ for every $x > 1$  such that $\sum_{\ell \geq 1}{\mathcal{G}_{\ell} \, x^{\ell} } < \infty$\,,
\item[(ii)] $\mathcal{G}^{-4} \mathbbm{1} \leq h \leq \mathcal{G}^{4} \, \mathbbm{1}$ and $\nu(h) = 1$.
\end{enumerate}
\end{Theo}

\begin{proof}
Let us consider the set 
\[ \mathcal{C} := \overline{\operatorname{conv}}{\{ L^{n}(\mathbbm{1}) \colon n \in \mathbb{N} \}}\,, \] 
which is clearly convex, closed and invariant by $L$, i.e. $L(\mathcal{C}) \subset \mathcal{C}$. Applying Corollary \ref{Coro:Lbounded} we deduce that  $ \mathcal{C} \subset \mathcal{A}_{\mathbb{N}}(x)$ whenever $x>1$ satisfies $\sum_{\ell \geq 1}{\mathcal{G}_{\ell} \, x^{\ell}} < \infty$ and that for every $Q \in \mathcal{C}$ and $\ell \geq 1$
\[ \nu(Q) = 1\,,\quad  \|Q\| \leq \mathcal{G}^{4}\,, \quad \| L^{n}(Q)\|_{\ell} \leq 2 \mathcal{G}^{3} \, \mathcal{G}_{\ell}\quad \mbox{ and } \quad  \mathcal{G}^{-4}  \, \mathbbm{1} \, \leq \, Q \, \leq \, \mathcal{G}^{4} \, \mathbbm{1}\,.\]
In particular, $\mathcal{C}$ is bounded and satisfies 
\begin{equation}\label{equa:auxfFixedPointHaux1}  
C \subset \mathcal{A}_{[1,\ell]} \, + \, (2 \, \mathcal{G}^{3} \, \mathcal{G}_{\ell}) \, \mathbb{B}_{\mathcal{A}_{\mathbb{N}}}\,
\end{equation}
where $\mathbb{B}_{\mathcal{A}_{\mathbb{N}}}$ is the closed unit ball of $\mathcal{A}_{\mathbb{N}}$. Since $\mathcal{A}_{[1,\ell]}$ is finite-dimensional, the set $\mathcal{C}$ is totally bounded, and thus compact. Schauder's fixed point theorem now applies to $L|_{\mathcal{C}}:\mathcal{C} \longrightarrow \mathcal{C}$ and let us to deduce the existence of $h \in C$ with $L(h) = h$. 
\end{proof}

\subsection{Convergence of $L^{n}$}

\begin{Theo}\label{Theo:FlatnessStrong}
Let $x >  1$ with $\sum_{\ell}{\mathcal{G}_{\ell} \,  x^{\ell}} < \infty$. Then, there exists an absolute constant $C_{x} > 0$ with the following property: for every $a>0$ there is $N = N(x,a)$ such that for all $n \geq N$ and $Q \in \mathcal{A}_{\mathbb{N}}$ positive and invertible:
\[ \mbox{ if }  \quad \| |Q|\|_{1,x} \| Q^{-1}\| \leq a \quad \mbox{ then } \quad\|| L^{n}(Q)|\|_{1,x} \,  \| L^{n}(Q)^{-1}\|^{-1} \, \leq \, C_{x} \,. \]
\end{Theo}

\begin{proof}
Let us fix $x>1$ as above. Given $3 \leq r \in \mathbb{N}$ and $a >0$, we are going to show that for every $n \geq N:=3r$ and every $Q \in \mathcal{A}_{\mathbb{N}}$ positive and invertible with $\|| Q|\|_{1,x} \, \| Q^{-1}\| \leq a$, it holds that
\begin{equation}\label{equa:FlatnessStrongAux0} 
\|| L^{n}(Q) |\|_{1,x} \, \| L^{n}(Q)^{-1}\| \, \leq K_{x}\,  (\, 1 + a \, x^{-r} \,) 
\end{equation}
for some constant $K_{x}$ independent of $a$ and $r$. It is clear that the theorem follows from this statement by taking $r = r(a,x)$ large enough. To prove the statement, let us first write $n$ in the form $n = 2r + r'$ with $r' \geq r$. On the one hand, we have by Theorem \ref{Theo:PropertiesLn}.(iv) 
\[ \| L^{2r+r'}(Q) \|  \, \| L^{2r+r'}(Q)^{-1}\| \, \leq \, \mathcal{G}^{4} \, \| L^{2r}(Q) \|  \, \| L^{2r}(Q)^{-1}\|\,.  \]
On the other hand, for each $ \ell \geq 1$ we can apply Theorem \ref{Theo:PropertiesLn}.(v) to obtain
\begin{align*}
\| L^{2r + r'}& (Q) \|_{\ell}  \, \| L^{2r + r'}(Q)^{-1}\|  =  \| L^{r'} (L^{2r}(Q))\|_{\ell} \, \| L^{r'} (L^{2r} (Q))^{-1}\| \\[3mm]
& \leq 2 \mathcal{G}^{3} \, \mathcal{G}_{\ell} \, \| L^{2r}(Q)\| \, \| L^{2r}(Q)^{-1}\| + \mathcal{G}^{4} \, \| L^{2r}(Q)\|_{r'+\ell} \, \| L^{2r}(Q)^{-1}\|\,.
\end{align*}
Combining the last two inequalities we get that
\begin{equation}\label{equa:FlatnessStrongAux1}
\begin{split}
\|| L^{n}(Q)|\|_{1,x} \, \| L^{n} (Q)^{-1}\|  \, & = \, \| L^{n}(Q) \|  \, \| L^{n}(Q)^{-1}\| \\[3mm]
& \hspace{20mm}  +  \, \sum_{\ell \geq 1} \, \| L^{n}(Q) \|_{\ell}  \, \| L^{n}(Q)^{-1}\| \, x^{\ell}\\[2mm]
& \, \leq \, 2 \, \mathcal{G}^{4}  \, \Big( 1 + \sum_{\ell \geq 1} \mathcal{G}_{\ell} \, x^{\ell} \Big) \, \| L^{2r}(Q)\| \, \| L^{2r}(Q)^{-1}\| \\[2mm]
& \hspace{20mm} + \mathcal{G}^{4} \, \sum_{\ell \geq 1}{\| L^{2r}(Q)\|_{r'+\ell} \, \| L^{2r}(Q)^{-1} \| \, x^{\ell}}\,.\\
\end{split}
\end{equation}
To bound the first summand we can apply Theorem \ref{Theo:Flatness} with $n= 2r$ and $\ell = r$
\begin{equation}\label{equa:FlatnessStrongAux2}
\begin{split} 
\| L^{2r}(Q)\| \, \| L^{2r}(Q)^{-1}\| \, & \leq \, \mathcal{G}^{4} \left( 1 + 2 \mathcal{G}^{3} \mathcal{G}_{r} \| Q\| \, \| Q^{-1}\| + 2 \mathcal{G}^{4} \, \| Q\|_{r} \, \| Q^{-1}\|\right)\\[2mm]
& \leq  \, \mathcal{G}^{4} \, \left( 1 + 2 \,a \, \mathcal{G}^{3} \, \mathcal{G}_{r}  + 2 \,a\, \mathcal{G}^{4} \, x^{-r} \right)\\[2mm]
& \leq  \mathcal{G}^{4} \, (2 \, \mathcal{G}^{3} \, \mathcal{G}_{r} \,x^{r} +2 \, \mathcal{G}^{4} ) \, \left( 1 + a \, x^{-r} \right)\,.
\end{split}
\end{equation}
A similar idea works for the second summand if we first apply Theorem \ref{Theo:PropertiesLn}.(v)
\begin{equation}\label{equa:FlatnessStrongAux3}
\begin{split}
\sum_{\ell \geq 1} \|&  L^{2r} (Q)\|_{r'+\ell}   \, \| L^{2r}(Q)^{-1} \| \, x^{\ell} \\[2mm] 
& \leq \, 2 \, \mathcal{G}^{3}\, \sum_{\ell \geq 1} \, \mathcal{G}_{r'+\ell} \, x^{\ell} \, \| Q\| \, \| Q^{-1}\| \, + \,  \mathcal{G}^{4} \, \sum_{\ell \geq 1} \| Q\|_{2r+r'+\ell} \, \| Q^{-1}\| \, x^{\ell}\\[2mm]
& \leq \, 2  \, \mathcal{G}^{4} \, \Big( 1 + \sum_{\ell \geq 1}{\mathcal{G}_{\ell} \, x^{\ell}}\Big) \,a \, x^{-r}\,.  \
\end{split}
\end{equation}
Finally, applying \eqref{equa:FlatnessStrongAux2} and \eqref{equa:FlatnessStrongAux3} to \eqref{equa:FlatnessStrongAux1} we conclude that the statement \eqref{equa:FlatnessStrongAux0} holds.
\end{proof}

The previous result is the key ingredient for the following main theorem.

\begin{Theo}\label{Theo:mainConvergence}
Let $x>1$ with $\sum_{\ell \geq 1}{\mathcal{G}_{\ell} \, x^{\ell}} < \infty$. Then, there exist $K_{x}>0$ and $\delta_{x} > 0$ with the following property: for every $n \in \mathbb{N}$ and $Q \in \mathcal{A}_{\mathbb{N}}(x)$ 
\[ \|| L^{n}(Q) - \nu(Q) h |\|_{1,x} \leq K_{x} \, e^{- n \, \delta_{x}} \, \|| Q  |\|_{1,x}\,. \]
\end{Theo}

\begin{proof}
Let $C = C_{x}$ be the constant provided in Theorem \ref{Theo:FlatnessStrong}  (that we can assume to be greater than three) and let $N = N(x)$ be the corresponding number when applying the aforementioned theorem for $a = 2 C_{x}$. This means that if $Q \in \mathcal{A}_{\mathbb{N}}$ is positive, invertible and satisfies
\[ \|| Q|\|_{1,x} \, \| Q^{-1}\| \, \leq \, 2 C\,, \]
then, for every $n \geq N$
\[ \|| L^{n}(Q)|\|_{1,x} \, \leq \, C \, \| L^{n}(Q)^{-1}\|^{-1}  \]
and so using that $\nu(L^{n}(Q))  = \nu(Q)$ 
\[ \nu(Q) \, \leq \, \|| L^{n}(Q)|\|_{1,x}  \, \leq \, C \, \| L^{n}(Q)^{-1}\|^{-1} \leq C \, \nu(Q)\,.  \]
This means that the linear operator $\psi_{N}: \mathcal{A}_{\mathbb{N}} \longrightarrow \mathcal{A}_{\mathbb{N}}$ given by
\[ \psi_{N}(A) = L^{N}(A) - \frac{\nu(A)}{2C} \, \mathbbm{1}\,, \quad A \in \mathcal{A}_{\mathbb{N}} \]
satisfies that $\psi_{N}(Q)$ is positive and invertible with
\[ \| \psi_{N}(Q) ^{-1} \|^{-1} \geq \| L^{N}(Q)^{-1}\|^{-1} - \frac{\nu(Q)}{2C} \, \geq \, \frac{\nu(Q)}{2C} \geq \frac{\| L^{N}(Q)^{-1}\|^{-1}}{2} \,,  \]
and moreover
\[
\| \psi_{N}(Q)\| \, \leq \, \| L^{N}(Q)\| \quad\mbox{ and } \quad \| \psi_{N}(Q)\|_{\ell} \, = \, \| L^{N}(Q)\|_{\ell} \quad (\ell \geq 1)\,.\\
\]
Altogether yields that the above $Q$ satisfies
\[ \|| \psi_{N}(Q) \||_{1,x} \, \| \psi_{N}(Q)^{-1} \| \, \leq \, 2 \, \|| L^{N}(Q) \||_{1,x} \, \| L^{N}(Q)^{-1} \| \, \leq \, 2 C\,. \]
Notice that we can iterate this process and get that  for every $k \in \mathbb{N}$ 
\[ \|| \psi_{N}^{k}(Q) \||_{1,x} \, \| \psi_{N}^{k}(Q)^{-1} \| \, \leq \, 2 C\,. \]
Consequently
\[ \|| \psi_{N}^{k}(Q) \||_{1,x} \, \leq \, 2  C\,  \| \psi_{N}^{k}(Q)^{-1} \|^{-1}  \, \leq \, 2 C  \nu( \psi_{N}^{k}(Q) ) = 2 C \, \left( 1 - \frac{1}{2C} \right)^{k}\,.   \]
Having this observation in mind we can now prove the result. We are going to distinguish several cases:\\

\noindent (1) Let $Q \in \mathcal{A}_{\mathbb{N}}(x)$ be a self-adjoint element with $\nu(Q) = 0$ and $\|| Q|\|_{1,x} \leq 1$. Write $Q_{1} = Q+ 2 \mathbbm{1}$ and $Q_{2} = 2 \mathbbm{1}$ which are both positive and invertible. Then 
\[ \||Q_{i} |\|_{1,x} \| Q_{i}^{-1}\| \, \leq \, \|| Q_{i} |\|_{1,x} \, \leq \, 2 + \|| Q |\|_{1,x} \, \leq \, 3\, \leq C \, \quad (i=1,2).  \]
We can then apply the previous observation and get that for every $k \in \mathbb{N}$
\[ \|| \psi_{N}^{k}(Q_{i}) |\|_{1,x} \leq 2 C \, \Big( 1 - \frac{1}{2C}\Big)^{k} \quad (i=1,2) \,. \]
But since $\nu(Q) = 0$, for every $k \in \mathbb{N}$
\begin{align*} 
\|| L^{kN}\big( Q \big) |\|_{1,x}  \, = \, \|| \psi_{N}^{k}(Q)  |\|_{1,x} \, \leq \,\||  \psi_{N}^{k}(Q_{1}) |\|_{1,x} + \|| \psi_{N}^{k}(Q_{2})|\|_{1,x}  \leq   4 C \, \Big( 1 - \frac{1}{2C} \Big)^{k}\,.
\end{align*}
If we take an arbitrary $n \in \mathbb{N}$ and write it as $n = k N + r$ with $k=[n/N]$   and $0 \leq r < N$, then we deduce from the previous inequality and the second part of Corollary \ref{Coro:Lbounded} that
\begin{align*} 
\|| L^{n}(Q) |\|_{1,x} \, = \, \|| L^{r} \big( L^{k N}(Q) \big) |\|_{x} \,  \leq \, 4 C \, \|| L^{r} |\|_{1,x} \, \Big( 1 - \frac{1}{2C} \Big)^{k} \, \leq \, 4\, \widetilde{C}_{x}  \, \Big( 1 - \frac{1}{2C} \Big)^{\frac{n}{N}}\,. 
\end{align*}
Observe that the constant $\widetilde{C}_{x}$ does only depend on $x$ since $0 \leq r < N$ and $N = N(x)$ was fixed at the beginning.\\

\noindent (2) Let $Q \in \mathcal{A}_{\mathbb{N}}(x)$ be a self-adjoint element. We deduce from case (1) applied to $Q' = Q - \nu(Q) h$ that
\begin{align*} 
\|| L^{n}(Q) - \nu(Q) h |\|_{1,x} \, & \leq \,  4 \, \widetilde{C}_{x} \, \Big( 1 - \frac{1}{2C} \Big)^{\frac{n}{N}}\ \|| Q - \nu(Q) h |\|_{1,x}\\[2mm]
& \leq 4 \, \widetilde{C}_{x} \, (1 + \|| h |\|_{1,x}) \, \Big( 1 - \frac{1}{2C} \Big)^{\frac{n}{N}}\ \|| Q  |\|_{1,x}\, 
\end{align*}
The previous inequality extends straightforwardly to any $Q \in \mathcal{A}_{\mathbb{N}}(x)$ as it can be written as the difference of two self-adjoint operators in $\mathcal{A}_{\mathbb{N}}(x)$ with norm $\||\cdot |\|_{1,x}$ less than or equal to that of $Q$.
\end{proof}

\begin{Coro}\label{Coro:mainConvergence}
Let $x>1$ with $\sum_{\ell \geq 1}{\mathcal{G}_{\ell} \, x^{\ell}} < \infty$. Then, there exist $\widetilde{K}_{x}>0$ and $\delta_{x} > 0$ with the following property: for every $m, n \geq 0$ and $Q \in \mathcal{A}_{\mathbb{N}}$ it holds that
\[ \|| L^{n+m}(Q) - \nu(Q) h |\|_{1,x} \leq \widetilde{K}_{x} \, e^{- n \, \delta_{x}} \, \|| Q   |\|_{1+m,x}\,. \]
In particular, if $Q \in \mathcal{A}_{[1,m]}$ then 
\[ \|| L^{n+m}(Q) - \nu(Q) h |\|_{1,x} \leq \widetilde{K}_{x} \, e^{- n \, \delta_{x}} \, \| Q   \|\,. \]
\end{Coro}

\begin{proof}
Combining Theorem \ref{Theo:mainConvergence} and Corollary \ref{Coro:Lbounded} we deduce that
\begin{align*}
\|| L^{n+m}(Q) - \nu(Q) h |\|_{1,x} & \, = \, \|| L^{n}(L^{m}(Q)) - \nu(L^{m}(Q)) h |\|_{1,x}\\[2mm]  
& \, \leq \, K_{x} \, e^{- n \, \delta_{x}} \, \|| L^{m}(Q)   |\|_{1,x}\\[2mm]
& \, \leq \, \widetilde{K}_{x} \, e^{- n \, \delta_{x}} \, \|| Q   |\|_{m+1,x}\,.
\end{align*}
\end{proof}


\subsection{Exponential decay of correlations and convergence of states}

\begin{Theo}\label{Theo:ExponentialDecay}
Let $n, \ell \geq 1$ and $A \in \mathcal{A}_{[1,n]}$ and $B \in \mathcal{A}_{\mathbb{N}}$. Then,
\begin{equation}\label{equa:TheoExponentialDecay1}
\| L^{n}(A \, \tau_{n+\ell}(B)) - \tau_{\ell}(B) \, L^{n}(A) \| \leq 4 \mathcal{G}^{3}\, \mathcal{G}_{\ell} \, \| A\| \, \| B\| \,.\\[2mm] 
\end{equation}
In particular, there exist constants $K, \delta > 0$ independent of $n, \ell$ and the observables $A,B$  such that
\begin{equation}\label{equa:TheoExponentialDecay2}
|\nu(A \, \tau_{n+\ell}(B)) - \nu(A) \, \nu(\tau_{n+\ell}(B))|\,  \leq \, K \,  \| A \| \, \| B\| \, e^{- \delta \, \ell}\,.
\end{equation}
\end{Theo}

\begin{proof}
Let us fist notice that
\begin{equation}\label{equa:ExponentialDecayAux1}
\mathcal{L}_{(n,n+\ell)} \big(\, A \, \tau_{n+\ell}(B) \big) \, = \, \tau_{\ell}(B) \, \mathcal{L}_{(n,n+\ell)}(A)
\end{equation}
since the support of $\tau_{n+\ell}(B)$ is contained in $[1+n+\ell, \infty)$ and so it is disjoint from the other factors and the lattice sites where the partial trace acts. We can then estimate, using Theorem \ref{Theo:PropertiesLn}
\begin{align*}
\| \mathcal{L}_{n}(A\tau_{n+\ell}(B)) - \mathcal{L}_{(n,n+\ell)}(A\tau_{n+\ell}(B))\| \, & \leq \, \operatorname{tr}\big( e^{- H_{[1,n]}}\big)\, 2 \, \mathcal{G} \, \mathcal{G}_{\ell} \, \| A \, \tau_{n+\ell}(B)\|\,,\\[2mm]
\| \tau_{\ell}(B) \mathcal{L}_{n}(A) - \tau_{\ell}(B) \mathcal{L}_{(n,n+\ell)}(A)\| \, & \leq \, \operatorname{tr}\big( e^{- H_{[1,n]}}\big)\, 2 \, \mathcal{G} \, \mathcal{G}_{\ell} \, \| A \| \, \|\tau_{\ell}(B)\|\,.
\end{align*}
Combining these inequalities with \eqref{equa:ExponentialDecayAux1} we deduce
\[ \| \mathcal{L}_{n}(A\tau_{n+\ell}(B)) - \tau_{\ell}(B) \mathcal{L}_{n}(A) \| \, \leq \, \operatorname{tr}(e^{-H_{[1,n]}}) \, 4 \mathcal{G} \, \mathcal{G}_{\ell} \, \| A\| \, \| B\|\,. \]
Dividing by $\mu^{n}$ on both sides and using \eqref{equa:fixedStateAux} we conclude that \eqref{equa:TheoExponentialDecay1} hold. Let us focus now on the second statement. We can assume w.l.o.g. that $\nu(A) = 0$ replacing $A$ with $A - \nu(A) \mathbbm{1}$. Let $r := \lfloor \ell/2\rfloor$, then
\begin{align*} 
 |\nu(A \tau_{n+\ell}(B))| & = |\nu(L^{n+r}(A\tau_{n+\ell}(B)))| \leq \| L^{n+r}(A\tau_{n+\ell}(B))\|\\[2mm] 
 & \leq \, \| L^{n+r}(A\tau_{n+\ell}(B)) - \tau_{\ell}(B) \, L^{n+r}(A)\| + \| B\| \, \| L^{n+r}(A)\|\,.
 \end{align*}
 Applying \eqref{equa:TheoExponentialDecay1} and Corollary \ref{Coro:mainConvergence}, we have for suitable constants $ \widetilde{K} ,\widetilde{\delta} > 0$ independent of $n,r$
 \[ |\nu(AB)| \leq 2 \, \mathcal{G}^{3} \, \mathcal{G}_{\ell-r} \, \| A\| \, \| B\|  + \widetilde{K} \, \| A \|\, \| B\|  \, e^{-r \, \widetilde{\delta}} \]
 where $\mathcal{G}_{\ell - r} \leq \mathcal{G}_{r}$ since $\ell - r \geq r$, so using that $\mathcal{G}_{r}$ converges to zero exponentially fast in $r$ we conclude the result.
\end{proof}

\begin{Lemm}\label{Lemm:convergenceStates}
Let $x>1$ with $\sum_{\ell} \mathcal{G}_{\ell} \, x^{\ell} < \infty$. Then, there are constants $K_{x}, \delta_{x}>0$ such that for every $a,k \geq 1$ and every $Q \in \mathcal{A}_{\mathbb{N}}(x)$
\[ \left| \varphi^{[1,a+k]}(Q) - \nu(Q) \right| \, \leq \, K_{x} \, e^{- \delta_{x} a} \, \||Q |\|_{1 + k, x}\,.  \]
In particular, 
the sequence of states $\varphi^{[1,k]}$ converges pointwise on $\mathcal{A}_{\mathbb{N}}$ to the state $\nu$.
\end{Lemm}

\begin{proof}
We can write for every $1 \leq n < a':=a+k$
\[ \varphi^{[1,a']}(Q) \, = \, \frac{\varphi^{[1,a'-n]}\left( L_{(n,a')}(Q) \right)}{\varphi^{[1,a'-n]}\left( L_{(n,a')}(\mathbbm{1})\right)} \quad, \quad \nu(Q) = \frac{\varphi^{[1,a'-n]}\left( \nu(Q) h\right)}{\varphi^{[1,a'-n]}\left(  h\right)}\]
where the first identity already appeared in \eqref{equa:localGibbsL}. Thus
\[ \left| \varphi^{[1,a']}(Q) - \nu(Q) \right| \, \leq \, \big| \varphi^{[1,a']}(Q)\big| \, \frac{ \| L_{(n,a')}(\mathbbm{1}) -  h \|}{\varphi^{[1,a'-n]}(h)} \, + \frac{\| L_{(n,a')}(Q) -  \nu(Q) h \| }{\varphi^{[1,a'-n]}(h)}.\]
By Corollary \ref{Coro:Lbounded}.(ii) and Theorem \ref{Theo:mainConvergence}  
\begin{align*}
\| L_{(n,a')}(Q) - \nu(Q)h \| & \leq \| L_{(n,a')}(Q) - L^{n}(Q)\| + \| L^{n}(Q) - \nu(Q) h\|\\[2mm]
& \leq 2 \, \mathcal{G}^{3} \,  \mathcal{G}_{a'-n} \, \| Q\| + \widetilde{K}_{x} \, e^{-(n-k) \tilde{\delta}_{x}} \||Q |\|_{1+k, x}\,.
\end{align*}
Moreover, applying the state $\varphi^{[1,a'-n]}$ to Theorem \ref{Theo:fixedPointH}.(ii) we get $\varphi^{[1,a'-n]}(h) \geq \mathcal{G}^{-2}$. Therefore, combining these estimates
\[  \left| \varphi^{[1,a']}(Q) - \nu(Q) \right| \, \leq \,  4 \, \mathcal{G}^{5} \, \mathcal{G}_{a'-n}  \, \| Q\| + \mathcal{G}^{2} \, \widetilde{K}_{x} \, e^{-(n-k) \, \tilde{\delta}_{x}} \, \||Q |\|_{1+k, x}\,. \]
Finally, if we take $n := \lfloor a/2\rfloor + k$, we have that $n-k = \lfloor a/2\rfloor$ and $a'-n \geq \lfloor a/2\rfloor $, so that $\mathcal{G}_{a'-n} \leq \mathcal{G}_{\lfloor \ell/2 \rfloor}$. Since the latter converges to zero exponentially fast in $a$, we conclude the result.
\end{proof}

\begin{Theo}\label{Theo:KMSstate}
There is a translation invariant state $\psi$ on $\mathcal{A}_{\mathbb{N}}$ such that
\begin{equation}\label{equa:KMSstateaux}  
\lim_{k \rightarrow \infty} \, \| \nu \circ \tau_{k} - \psi \| \, e^{\delta k}  \, = \, 0 \quad \mbox{ for some $\delta > 0$}  \,. 
\end{equation}
Moreover, the (unique) extension of $\psi$ to a translation invariant state on $\mathcal{A}_{\mathbb{Z}}$, that we also denote by $\psi$, satisfies that there are constants $K, \delta>0$ such that\\[-2mm]
\begin{enumerate}\itemsep0.7em
\item[(i)] For every $j \in \mathbb{Z}\,, \, k \in \mathbb{N}$ and every $A \in \mathcal{A}_{(-\infty, j]}$ and $B \in \mathcal{A}_{[j+k, \infty)}$
\[ |\psi(AB) - \psi(A) \, \psi(B)| \, \leq \, \| A\| \, \| B\| \, K \, e^{ - \delta k}\,. \]
\item[(ii)] For every $Q \in \mathcal{A}_{[a,b]}$ and every $k \in \mathbb{N}$
\[ \left| \psi(Q) - \varphi^{[a-k,b+k]}(Q) \right| \leq K e^{- \delta k} \,.\\[2mm] \]
\end{enumerate}
In other words, $\psi$ is the infinite-volume KMS state associated to $\Phi$ at temperature one and satisfies exponential decay of correlations.
\end{Theo}

\begin{proof}
Let $n,N \in \mathbb{N}$ and $Q \in \mathcal{A}_{\mathbb{N}}$. Using Theorems \ref{Theo:ExponentialDecay}  and \ref{Theo:mainConvergence} we have
\begin{align*}
\left| \nu\left(\tau_{n+N}(Q) \right) - \nu\left( \tau_{N}(Q) h\right) \right| \, & = \, \left| \nu\big(L^{n}\big( \tau_{N +n}(Q)\big) \big) - \nu\big(\tau_{N}(Q)h\big) \right|\\[2mm]
& \, \leq \| L^{n}\big( \tau_{N+n}(Q) \big) - \tau_{N}(Q) \, L^{n}(\mathbbm{1}) \| \, + \, \| Q\| \, \| L^{n}(\mathbbm{1}) - h\|\\[2mm]
& \, \leq \, \left( 4 \, \mathcal{G} \, \mathcal{G}_{N} \,  +  \, K_{x} \, e^{- n \, \delta_{x}} \right) \, \| Q\|\,.
\end{align*}
As a consequence, $(\nu \circ \tau_{k})_{k}$ is a Cauchy sequence convergent to a state $\psi$ over $\mathcal{A}_{\mathbb{N}}$ satisfying moreover
 \[ \| \nu \circ \tau_{n+N} - \psi\| \, \leq \, 8 \, \mathcal{G} \, \mathcal{G}_{N} + 2 K_{x} e^{- n \, \delta_{x}}\,. \]
It is clear from the previous expression that $\psi \circ \tau = \psi$, that is $\psi$ is translation invariant. Let us now check statements (i) and (ii).\\[-2mm]

\noindent (i) Let $A \in \mathcal{A}_{(-\infty, j]}$ and $B \in \mathcal{A}_{[j+k, \infty)}$ as above. We can assume w.l.o.g. that $A$ is a local observable. Then, there exist $n_{0} \in \mathbb{N}$ such that $\tau_{n}(A), \tau_{n}(B) \in \mathcal{A}_{\mathbb{N}}$ for every $n \geq n_{0}$. We have by Theorem \ref{Theo:ExponentialDecay} that there are constants $K, \delta > 0$ independent of $k,j,n_{0},n$ such that
\[ |\nu(\tau_{n}(A) \tau_{n}(B)) - \nu(\tau_{n}(A)) \nu(\tau_{n}(B))| \leq K  e^{- \delta k} \,\| A\|\, \| B\|\,. \]
Taking limit on $n$ and using \eqref{equa:KMSstateaux}  we conclude the result.\\[-2mm]

\noindent (ii) Let $Q \in \mathcal{A}_{[a,b]}$ and $k \in \mathbb{N}$. Then
\begin{align*}
\left| \varphi^{[a-k,b+k]}(Q) - \psi(Q) \right| \, & =\, \left| \varphi^{[1,1+b-a+2k]}(\tau_{1+k-a}(Q)) - \psi(Q)\right|\\[2mm]
& \leq \, \left| \varphi^{[1,1+b-a+2k]}(\tau_{1+k-a}(Q)) - \nu(\tau_{1+k-a}(Q))\right| +\\[2mm]
& \hspace{4cm} + \left| \nu(\tau_{1+k-a}(Q)) - \psi(Q)\right|\,.
\end{align*}
The first summand can be estimated using Lemma \ref{Lemm:convergenceStates} 
\begin{align*} 
\left| \varphi^{[1,1+b-a+2k]}(\tau_{1+k-a}(Q)) - \nu(\tau_{1+k-a}(Q))\right| \leq K_{x} \, e^{-(k-1) \,  \delta_{x}} \|| \tau_{1+k-a}(Q) |\|_{x, 2+b-a+k}
 \end{align*}
and using that $\tau_{1+k-a}(Q)$ has support in $[1+k,1+b-a+k]$, so that 
\[ \|| \tau_{1+k-a}(Q) |\|_{x, 2+b-a+k} = \|  \tau_{1+k-a}(Q) \| = \| Q\|\,. \]
The second summand can be estimated replacing $\psi(Q) = \psi(\tau_{1-a}(Q))$ by translational invariance and using \eqref{equa:KMSstateaux}.
\end{proof}

\newpage

\section{Spectral gap problem for 2D PEPS}
\label{sec:2DPEPS}

\subsection{Basics}

Let us consider the two-dimensional lattice $\mathbb{Z}^{2}$ and denote by $\mathcal{E}$ its set of edges. Fixed $d, D \in \mathbb{N}$, we consider at each lattice point $v \in \mathbb{Z}^{2}$ a linear operator or tensor 

\begin{figure}[h]
\begin{tikzpicture}[scale=0.3]
\draw (-6.1,0) node{$T_{v} \,  : (\mathbb{C}^{D})^{\otimes 4}  \,\,  \longrightarrow \,\,  \mathbb{C}^{d}$};

\draw (0,-4) node{$T_{v} \, = \, \sum\limits_{k=1}^{d} \,\, \sum\limits_{j_{1}, j_{2}, j_{3}, j_{4}=1}^{D} T_{j_{1}, j_{2}, j_{3}, j_{4}}^{k} \ket{k} \bra{j_{1} j_{2} j_{3} j_{4}}$}; 

\begin{scope}[xshift=-22cm, yshift=-2.5cm]
\draw[ thick] (0,-3) -- (0,3) ;  
\draw[thick] (-3,0) -- (3,0);  
\draw[fill=black!20] (-0.9,-0.9) rectangle (0.9,0.9);  

\draw (0,0) node{$k$}; 
\draw (2.5,1) node{$j_1$}; 
\draw (-1,2.5) node{$j_2$}; 
\draw (-2.5,-1) node{$j_3$}; 
\draw (1,-2.5) node{$j_4$}; 
\end{scope}
\end{tikzpicture}
\end{figure}

\noindent Here, $\mathcal{H}_{v} \equiv \mathbb{C}^{d}$ is the \emph{physical} space associated to $v$ and each $\mathbb{C}^{D}$ is the \emph{virtual} space corresponding to an edge $e$ incident to $v$. For each finite region $\mathcal{R} \subset \mathbb{Z}^{2}$ let us denote by $\mathcal{E}_{\mathcal{R}}$ the set of edges conecting vertices contained in $\mathcal{R}$, and by $\mathcal{E}_{\partial \mathcal{R}}$ the edges simultaneously incident to $\mathcal{R}$ and $\mathbb{Z}^{2} \setminus \mathcal{R}$. We can then associate to $\mathcal{R}$ the linear operator
\[ \mbox{$\bigotimes_{v \in \mathcal{R}}$}T_{v} \,: \, \mbox{$\bigotimes_{e \in \mathcal{E}_{\mathcal{R}}}$} \left( \mathbb{C}^{D} \otimes \mathbb{C}^{D}\right)  \otimes  \mbox{$\bigotimes_{e \in \mathcal{E}_{\partial \mathcal{R}}}$} \mathbb{C}^{D}\,\,   \longrightarrow \,\, \mbox{$\bigotimes_{v \in \mathcal{R}}$} \mathbb{C}^{d}. \]
If we set on each edge $e \in \mathcal{E}_{\mathcal{R}}$ a maximally entangled state
\[ \ket{\omega_{e}} = \frac{1}{\sqrt{D}} \sum_{j=1}^{D}{\ket{j\, j}} \in \mathbb{C}^{D} \otimes \mathbb{C}^{D}\,, \]
then, we get a linear operator from the virtual space $\mathcal{H}_{\partial \mathcal{R}} := \bigotimes_{e \in \mathcal{E}_{\partial \mathcal{R}}}{\mathbb{C}^{D}}$ into the bulk physical space $\mathcal{H}_{\mathcal{R}} := \bigotimes_{v \in \mathcal{R}}{\mathbb{C}^{d}}$
\[ T_{\Lambda}: \mathcal{H}_{\partial \mathcal{R}} \, \longrightarrow \, \mathcal{H}_{\mathcal{R}}\,,\quad \ket{X} \, \longmapsto \, \mbox{$\bigotimes_{v \in \mathcal{R}}{T_{v}}$} \Big( \mbox{$\bigotimes_{e \in \mathcal{E}_{\mathcal{R}}}$} \ket{\omega_{e}} \otimes \ket{X}\Big)\,. \]

\begin{figure}[h]
\begin{tikzpicture}[scale=0.3]


\begin{scope}[scale=1.5, xshift=6cm,yshift=-0.5cm]
\draw[step=1cm,black] (-1.9,-1.9) grid (4.9,2.9);

\foreach \x in {-1,...,4} \foreach \y in {-1,...,2}
\draw[fill=black!20] (\x-0.15,\y-0.15) rectangle (\x+0.15,\y+0.15);  

\end{scope}


\begin{scope}[scale=1.5, xshift=15.5cm,yshift=-0.5cm]
\draw[step=1cm,black] (-1.9,-1.9) grid (4.9,2.9);

\draw[fill=black!20] (-1.15,-1.15) rectangle (4.15,2.15);  

\draw (1.65,0.5) node{$\mathcal{H}_{\mathcal{R}}$}; 
\end{scope}


\begin{scope}[scale=1.5, xshift=25cm, yshift=-0.5cm]

\draw[step=1cm,black!20] (-1.9,-1.9) grid (4.9,2.9);

\draw 
  (-1.5,-1.5) --
  ++(6,0) --
  ++(0,4) --
  ++(-6,0) --
  cycle
  {};
  
  \foreach \x in {-1.5, 4.5} \foreach \y in {-1,...,2}
\draw[fill=black] (\x,\y) circle [radius = 0.1];   

  \foreach \x in {-1,...,4} \foreach \y in {-1.5,2.5}
\draw[fill=black] (\x,\y) circle [radius = 0.1];   

\draw (1.65,0.5) node{$\mathcal{H}_{\partial \mathcal{R}}$}; 

\end{scope}


\end{tikzpicture}

\caption{\footnotesize Contracted tensors on a finite rectangular region $\mathcal{R}$.}
\label{Fig:PEPS}
\end{figure}
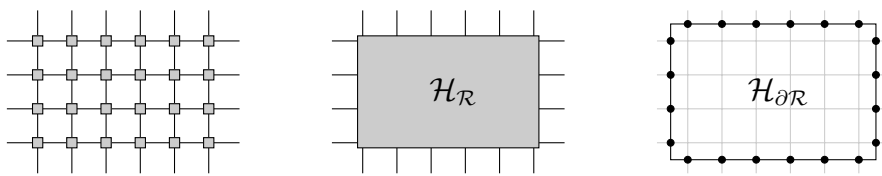

We say that the PEPS is \emph{injective} on $\mathcal{R}$ if the map $T_{\mathcal{R}}$ is injective. This property is somehow generic on regions large enough, and allows to construct a local Hamiltonian for which the PEPS is its unique ground state. Indeed, let us assume that every $T_{v}$ is injective. For each edge $e=(v_{1}, v_{2})$ of the lattice, let $h_{e}$  denote the orthogonal projection  onto the orthogonal complement of $\operatorname{Im} T_{\{ v_{1},v_{2}\}}$.  Then, the set of nearest-neighbor interactions $(h_{e})_{e \in \mathcal{E}}$ determines the \emph{parent Hamiltonian} of the PEPS. It satisfies that for every finite region $\mathcal{R}$ of the lattice, the groundspace of $H_{\mathcal{R}} = \sum_{e \in \mathcal{E}_{\mathcal{R}}}{h_{e}}$ coincides with $\operatorname{Im}T_{\mathcal{R}}$. Moreover, it is \emph{frustration-free}, namely that every ground state of $H_{\mathcal{R}}$ is a ground state of every local interaction term. See \cite{PeVeCiWo08} for a detailed exposition of these statements. \\

\noindent Finally, let us recall that the \emph{boundary state} of the PEPS on $\mathcal{R}$ is defined  as \footnote{We are following the definition given in \cite{KaLuPe19}. See there for a comparison with the definition given in \cite{CiPoScVe11}.}
\[ \rho_{\partial \mathcal{R}} := T_{\mathcal{R}}^{\dagger} T_{\mathcal{R}}\, \in \, \mathcal{B}(\mathcal{H}_{\partial \mathcal{R}}) \,. \]
Under the injectivity condition the boundary state is a full-rank positive matrix, and so it can be written as
\[  \rho_{\mathcal{R}}  =  \exp(2 \,G_{\partial \mathcal{R}})  \]
for some self-adjoint operator $G_{ \partial \mathcal{R}} \in \mathcal{B}(\mathcal{H}_{\partial \mathcal{R}})$ called the \emph{boundary Hamiltonian}. The choice of a factor two in the exponent is just convenient for later arguments.

\subsection*{Spectral gap and approximate factorization}

A main problem to tackle is finding conditions ensuring that the family of (parent) Hamiltonians $(H_{\mathcal{R}})_{\mathcal{R}}$ where $\mathcal{R}$ runs over all finite rectangles is \emph{gapped}, namely if 
\[ \inf_{\mathcal{R}}{\gamma(H_\mathcal{R})} > 0 \]
where $\gamma(H_{\mathcal{R}})$ is the difference between the two smallest eigenvalues of $H_{\mathcal{R}}$. This issue is related to the correlation properties in the bulk of the system, and the latter are connected to the locality features of the boundary states and Hamiltonians \cite{CiPoScVe11}. This relation was formalized in \cite{KaLuPe19} as a sufficient condition on the boundary states for the parent Hamiltonian to be gapped. We have to introduce some notation prior to the formal statement. 
\begin{Defi}
Given a length parameter $\ell \in \mathbb{N}$, let us say that three given rectangles $A,B, C$ are \emph{$\ell$-admissible} if they are adjacent, $B$ shields $A$ from $C$, and the width of $B$ is at least $4\ell$, as portrayed in Figure \ref{Fig:rectangulo0}. 
\end{Defi}

\begin{figure}[h]
\begin{tikzpicture}

\begin{scope}[scale=0.4]
\draw[step=1cm,black] (-4.9,-2.9) grid (5.9,2.9);

\foreach \x in {-4,...,5} \foreach \y in {-2,...,2}
\draw[fill=gray] (\x-0.15,\y-0.15) rectangle (\x+0.15,\y+0.15);  


\draw[fill=black!20] (-4-0.2,-2-0.2) rectangle (-2+0.2,2+0.2);

\draw[fill=black!20] (-1-0.2,-2-0.2) rectangle (2+0.2,2+0.2);

\draw[fill=black!20] (3-0.2,-2-0.2) rectangle (5+0.2,2+0.2);

\draw (-3,0) node{$A$}; 
\draw (0.5,0) node{$B$}; 
\draw (4,0) node{$C$}; 

\end{scope}

\begin{scope}[scale=0.4, xshift=13cm]

\draw[step=1cm,black!20] (-4.9,-2.9) grid (5.9,2.9);

\draw (-4.5,-2.5) -- (5.5,-2.5) -- (5.5,2.5) -- (-4.5,2.5) -- cycle;
  
\draw (-1.5,2.5) -- (-1.5, -2.5);
  \draw (2.5,2.5) -- (2.5, -2.5);
  
  \foreach \x in {-4.5, -1.5, 2.5, 5.5} \foreach \y in {-2,...,2}
\draw[fill=black] (\x,\y) circle [radius = 0.1];   

  \foreach \x in {-4,...,5} \foreach \y in {-2.5,2.5}
\draw[fill=black] (\x,\y) circle [radius = 0.1];   

\draw (-3,0) node{$A$}; 
\draw (0.5,0) node{$B$}; 
\draw (4,0) node{$C$}; 

\end{scope}

\begin{scope}[scale=0.4, xshift=26cm]

\draw (-4.5,-2.5) -- (5.5,-2.5) -- (5.5,2.5) -- (-4.5,2.5) -- cycle;
  
\draw (-1.5,2.5) -- (-1.5, -2.5);
  \draw (2.5,2.5) -- (2.5, -2.5);

\draw (-3,0) node{$A$}; 
\draw (0.5,0) node{$B$}; 
\draw (4,0) node{$C$}; 

\draw[dashed] (-0.5,2.5) -- (-0.5, -2.5);
  \draw[dashed] (0.5,2.5) -- (0.5, 0.5);
    \draw[dashed] (0.5,-0.5) -- (0.5, -2.5);
\draw[dashed] (1.5,2.5) -- (1.5, -2.5);
  
  
  \draw[<->] (-1.4,-2.3) -- (-0.6,-2.3);
  \draw[<->] (-0.4,-2.3) -- (0.4,-2.3);
  \draw[<->] (0.6,-2.3) -- (1.4,-2.3);
  \draw[<->] (1.6,-2.3) -- (2.4,-2.3);

\draw (-1,-1.7) node{$\ell$};
\draw (0,-1.7) node{$\ell$};
\draw (1,-1.7) node{$\ell$};
\draw (2,-1.7) node{$\ell$};

\end{scope}


\pgfmathsetmacro{\base}{3.5}
\pgfmathsetmacro{\altu}{1.6}

\begin{scope}[scale=1,xshift=6cm,yshift=-4cm]

\begin{scope}[ultra thick]

\draw (-\base,\altu) -- ({\base},\altu);   
\draw (-\base,-\altu) -- ({\base},-\altu);   
\draw (-\base,\altu) -- (-\base,-\altu);   
\draw (\base,-\altu) -- (\base,\altu); 

\draw (-{\base/2},\altu) -- (-{\base/2},-\altu);   
\draw ({\base/2},\altu) -- ({\base/2},-\altu);   

\end{scope}

\draw [dashed] (-{\base/4},\altu) -- (-{\base/4},-\altu);   
\draw [dashed] (0,\altu) -- (0,-\altu);   
\draw [dashed] ({\base/4},\altu) -- ({\base/4},-\altu);   


\begin{scope}[scale=1.1]
\draw [|-|] (-{\base/4},\altu) -- (-{\base}+0.1,\altu) -- (-{\base}+0.1,-\altu) -- (-{\base/4},-\altu);
\draw (-{\base-0.1},0) node{$a$}; 
\end{scope}

\begin{scope}[scale=1.1]
\draw [|-|] ({\base/4},\altu) -- ({\base-0.1},\altu) -- ({\base-0.1},-\altu) -- ({\base/4},-\altu);
\draw ({\base+0.1},0) node{$c$}; 
\end{scope}

\begin{scope}[scale=1.1]
\draw [|-|] (-{\base/4+0.1},\altu) -- ({\base/4-0.1},\altu);
\draw (0,{\altu+0.2}) node{$z$}; 
\end{scope}

\begin{scope}[scale=1.1]
\draw [|-|] (-{\base/4+0.1},-\altu) -- ({\base/4-0.1},-\altu);
\draw (0,{-\altu-0.2}) node{$z$}; 
\end{scope}

\begin{scope}[scale=0.9]
\draw [|-|] ({-\base/4-0.2},\altu) -- ({-\base/2},\altu) -- ({-\base/2},-\altu) -- ({-\base/4-0.2},-\altu);
\draw ({-\base/2+0.2},0) node{$\alpha$}; 
\end{scope}

\begin{scope}[scale=0.9]
\draw [|-|] ({\base/4+0.2},\altu) -- ({\base/2},\altu) -- ({\base/2},-\altu) -- ({\base/4+0.2},-\altu);
\draw ({\base/2-0.2},0) node{$\gamma$}; 
\end{scope}

\begin{scope}[scale=0.9]
\draw [|-|] (-{\base/4},\altu) -- (-0.1,\altu); 
\draw ({-0.5},{\altu-0.2}) node{$x$}; 
\draw [|-|] (-{\base/4},-\altu) -- (-0.1,-\altu); 
\draw ({-0.5},{-\altu+0.2}) node{$x$}; 

\draw [|-|]  (0.1,\altu) -- ({\base/4},\altu); 
\draw ({0.5},\altu-0.2) node{$y$}; 
\draw [|-|]  (0.1,-\altu) -- ({\base/4},-\altu); 
\draw ({0.5},-\altu+0.2) node{$y$}; 
\end{scope}

\end{scope}

\end{tikzpicture}
\caption{\small $A,B,C$ are $\ell$-admissible rectangles. At the bottom, notation introduced for boundary segments .}
\label{Fig:rectangulo0}
\end{figure}
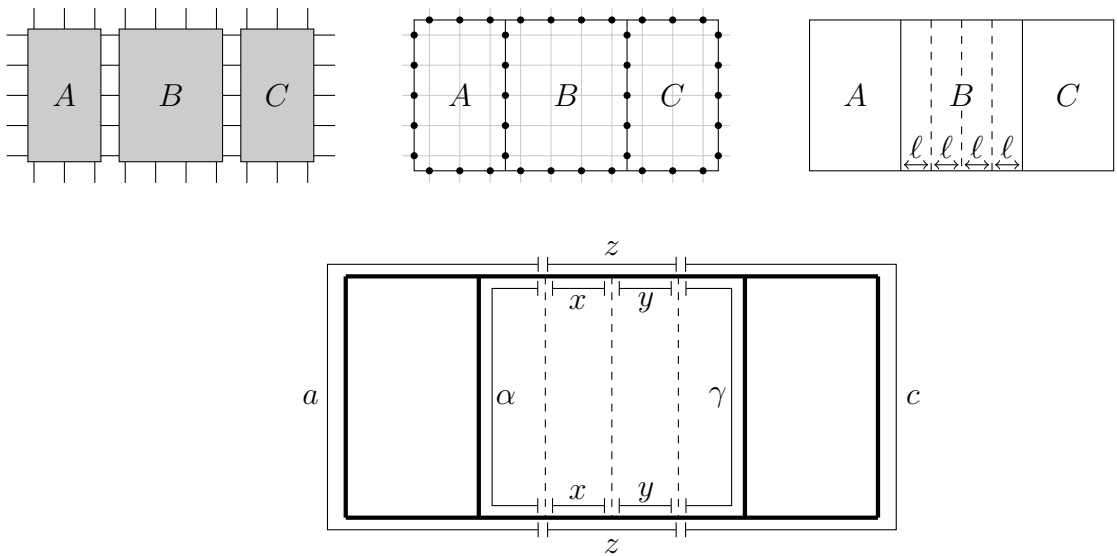

In this setup, let us introduce in Figure \ref{Fig:rectangulo0} some further notation for the boundary regions of $ABC$ that we will use in the following subsections. The horizontal sides of $B$ are split into four segments of length greater than the length scale $\ell$. Notice that $z$ and so $x,y$ are at distance greater than $\ell$ from $\partial A \cup \partial C$, while $a$  and $\alpha$ (as well as $c$ and $\gamma$) overlap on segments with length greater than $\ell$, by definition.

\begin{Defi}
Let us say that a family of positive and invertible observables $(\rho_{\partial \mathcal{R}})_{\mathcal{R}}$, where $\mathcal{R}$ runs over all finite rectangles, is \emph{approximately factorizable} if there is a positive decreasing function $ \ell \longmapsto \varepsilon(\ell)$ with polynomial decay, i.e. $\varepsilon(\ell) = \operatorname{poly}(1/\ell)$ fulfilling that for every $\ell$-admissible rectangles $ABC$ there exist
invertible observables $\Delta_{az}, \Delta_{zc}, \Upsilon_{\alpha z}, \Upsilon_{z \gamma}$ on $\partial A \cup \partial B \cup \partial C$ (subscripts indicate the support according to Figure \ref{Fig:rectangulo0}) such that the full-rank matrices
\[ \sigma_{\partial ABC}:= \Delta_{zc} \Delta_{az}\,, \,\,\, \sigma_{\partial AB}:=\Upsilon_{z \gamma} \Delta_{az}\,,\,\,\, \sigma_{\partial BC}:=\Delta_{zc} \Upsilon_{\alpha z}\,, \,\,\, \sigma_{\partial B} = \Upsilon_{z \gamma} \Upsilon_{\alpha z}\, \]
satisfy
\[ \| \rho_{\partial \mathcal{R}}^{1/2} \, \sigma_{\partial \mathcal{R}}^{-1} \, \rho_{\partial \mathcal{R}}^{1/2} \, - \,  \mathbbm{1}\| < \varepsilon(\ell)\quad \mbox{for every } \mathcal{R} \in \{ABC, AB, BC, B \}\,. \]
\end{Defi}

If the boundary states of the above injective PEPS satisfy this property, then the family of parent Hamiltonians $(H_{\mathcal{R}})_{\mathcal{R}}$ where $\mathcal{R}$ runs over all finite rectangles in $\mathbb{Z}^{2}$ is gapped  \cite{KaLuPe19}. Moreover, the authors show that this approximate factorization property is satisfied if the boundary Hamiltonians posses nice locality and compatibility conditions.

\begin{Defi}[Locality and homogeneity]
Let us decompose each term of the family of boundary Hamiltonians $(G_{\partial \mathcal{R}})_{\mathcal{R}}$ as a sum of local interactions
\[ G_{\partial \mathcal{R}}= \sum_{X \subset \partial \mathcal{R}}{g_{X}^{\partial \mathcal{R}}}\,. \]
The rate of decay of this decomposition is defined as the sequence
\begin{equation*} 
\Omega_{k} := \, \mbox{$\sup_{\mathcal{R}}$} \,\,\, \mbox{$\sup_{x \in \partial \mathcal{R}}$} \,\, \sum \{  \, \| g_{X}^{\partial \mathcal{R}}\| \colon X \ni x \,, \,\, \operatorname{diam}_{\mathcal{R}}(X) \geq k \, \} \,\,, \,\, k \geq 0 
\end{equation*}
where the diameter is calculated in terms of the intrinsic distance $\operatorname{dist}_{\partial \mathcal{R}}$ of $\partial \mathcal{R}$ considered as a one-dimensional periodic lattice. In particular, we will say that the interactions decay exponentially if there is $\lambda > 0$ such that
\[ \| \Omega\|_{\lambda}:= \sum_{k=0}^{\infty} e^{\lambda k} \Omega_{k} < \infty\,. \]
In addition, given a positive sequence $\ell \longmapsto \eta(\ell)$ converging to zero as $\ell$ tends to infinity, we say that this decomposition is \emph{$\eta$-homogeneous} if for every pair of adjacent rectangles $AB$, the corresponding boundary Hamiltonians $G_{\partial AB}$ and $G_{\partial A}$ satisfy that for each $X \subset \partial{A} \setminus \partial{B}$\\ 
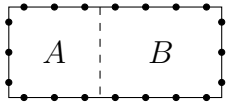
\begin{center}
\begin{tikzpicture}
\begin{scope}[scale=0.4]

\draw (-4.5,-1.5) -- (2.5,-1.5) -- (2.5,1.5) -- (-4.5,1.5) -- cycle;
  
\draw[dashed] (-1.5,1.5) -- (-1.5, -1.5);
  
  \foreach \x in {-4.5, 2.5} \foreach \y in {-1,...,1}
\draw[fill=black] (\x,\y) circle [radius = 0.1];   

  \foreach \x in {-4,...,2} \foreach \y in {-1.5,1.5}
\draw[fill=black] (\x,\y) circle [radius = 0.1];   

\draw (-3,0) node{$A$}; 
\draw (0.5,0) node{$B$}; 

\end{scope}

\draw (7,0) node{$\| g_{X}^{\partial AB} - g_{X}^{\partial A} \| \leq \eta(\operatorname{dist}_{\partial AB}(X, \partial B)) \, \left( \|g_{X}^{\partial AB} \|+ \| g_{X}^{\partial A}\|\right)$.};
\end{tikzpicture}
\end{center}
\end{Defi}

It is shown in \cite[Section 5]{KaLuPe19} that if the boundary Hamiltonians are finite-range (i.e. $\Omega_{k}=0$ if $k$ is larger than a fixed $r>0$) and the homogeneity condition holds for a controlling sequence $\eta$ that decays sufficiently (polynomially) fast, then the quasi-factorization condition holds. We aim to extend this result to interactions with exponential decay. Our method is analogous, based on the imaginary time locality estimate from previous sections, although simpler in the sense that we replace the use of expansional formulas for multiple product with just perturbation formulas.

\subsection*{Locality estimates on the boundary}

First we have to obtain some locality estimates on the imaginary-time evolution of an observable $Q \in \mathcal{B}(\mathcal{H}_{\partial \mathcal{R}})$ with respect to a boundary Hamiltonian $G_{\partial \mathcal{R}}$. We are going to call a supporting set $\Lambda_{0} \subset \partial \mathcal{R}$ \emph{admissible} if it consists of at most two connected segments. As usual, we denote
\[ \Lambda_{k}:= \{ x \in \partial{\mathcal{R}} \colon \operatorname{dist}_{\partial \mathcal{R}}(x, \Lambda_{0}) \leq k \}\,\, , k \geq 0\,. \]

\begin{Lemm}\label{Lemm:localityEstimates Boundary}
Let us consider a family of local boundary Hamiltonians $(G_{\partial \mathcal{R}})_\mathcal{R}$ with exponential decay $\| \Omega\|_{\lambda} < \infty$ for some $\lambda > 0$. Then, for every rectangle $\mathcal{R}$ and every observable $Q \in \mathcal{B}(\mathcal{H}_{\partial \mathcal{R}})$ with admissible support $\Lambda_{0}$ 
\begin{align}
\label{equa:localityQuasifactorization1} \| \Gamma_{\Lambda_{L}}^{s}(Q) - \Gamma_{\Lambda_{\ell}}^{s}(Q) \| & \leq \|Q \| e^{2 |s| \Omega_{0} |\Lambda_{0}| \, } \, e^{ 8 |s| \, \|\Omega\|_{\lambda} } \, e^{(8|s| \Omega_{0}- \lambda) \ell}\\[2mm]
\label{equa:localityQuasifactorization2} \| \Gamma_{\Lambda_{L}}^{s}(Q)\| & \leq \|Q \| e^{2 |s| \Omega_{0} |\Lambda_{0}| \, }  \, e^{ 8 |s| \, \|\Omega\|_{\lambda} }\,.
\end{align}
whenever  $0 \leq \ell \leq L$ and $s \in \mathbb{C}$ satisfies $|s| \leq \lambda/(8\Omega_{0})$. 
\end{Lemm}

\begin{proof}
Since $\Lambda_{0}$ consists of at most two connected segments,  $|\Lambda_{k} \setminus \Lambda_{k-1}| \leq 4$ for every $k \geq 1$. Consequently, and following the notation of Theorem \ref{Theo:analyticLRgeneral}, we can bound for every $0 \leq j < k$ 
\[ W(j, k) \,\, \leq \,\, |\Lambda_{k} \setminus \Lambda_{k-1}|  \, \Omega_{k-j} \,\, \leq \,\, 4 \, \Omega_{k-j}\,.\]
Analogously to Theorem \ref{Theo:FundamentalEstimations}, for every $0 \leq \ell \leq L$ and $s \in \mathbb{C}$
\begin{align*}
\| \Gamma_{\Lambda_{L}}^{s}(Q) - \Gamma_{\Lambda_{\ell}}^{s}(Q) \| & \leq \|Q \| e^{2 |s| \Omega_{0} |\Lambda_{0}| \, } \, \sum_{k=\ell +1}^{L} e^{8 |s| \Omega_{0} k} \Omega_{k}^{\ast}(8 |s|)\\
\| \Gamma_{\Lambda_{L}}^{s}(Q)\| & \leq \|Q \| e^{2 |s| \Omega_{0} |\Lambda_{0}| \, } \, \sum_{k=0}^{L} e^{8 |s| \Omega_{0} k} \Omega_{k}^{\ast}(8 |s|)\,.
\end{align*}
Finally, we can argue as in the proof of \eqref{equa:localityEstimatesExponentialDecay1}-\eqref{equa:localityEstimatesExponentialDecay2} to obtain \eqref{equa:localityQuasifactorization1}-\eqref{equa:localityQuasifactorization2}. 
\end{proof}

\subsection*{Quasiperturbation formulas}

We consider the border of a rectangle $AB$ made with two adjacent rectangles $A$ and $B$. It is formed by two connected segments, namely $\partial A \setminus \partial B$ and $\partial B \setminus \partial A$, see Figure ~\ref{Fig:quasiPerturbation1}. Denote by $S_{1}$ the set containing the (four) extreme points of both segments, and for each $j > 1$
\[ S_{j}:=\{ v \in \partial \mathcal{R} \colon \operatorname{dist}_{\partial \mathcal{R}}(v, S_{1}) \leq j \}\,. \]

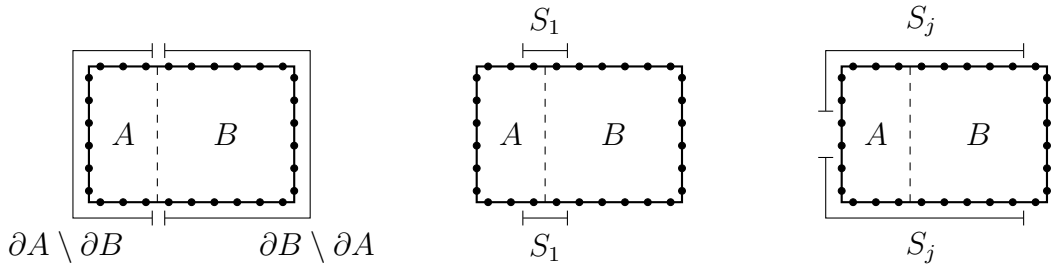
\begin{figure}[h]
\begin{tikzpicture}[scale=0.3]

\pgfmathsetmacro{\error}{0.7}

\begin{scope}[xshift=33cm]

\draw[thick] (-4,-3) -- (-4,3) -- (5,3) -- (5,-3) -- cycle;

\foreach \x in {-4,...,4} \foreach \y in {-3,3}
\draw[fill=black] ({0.5+\x},\y) circle [radius = 0.15];    

\foreach \x in {-4,5} \foreach \y in {-3,...,2} \draw[fill=black] (\x,{0.5+\y}) circle [radius = 0.15];    

\draw (-2.5,0) node{$A$}; 
\draw (2,0) node{$B$}; 


\draw[dashed] (-1, -3) -- (-1,3);  

\draw[|-|] ({-4- \error},{1}) -- ({-4-\error},{3+\error}) -- ({4},{3+\error}); 

\draw[|-|] ({-4- \error},{-1}) -- ({-4-\error},{-3-\error}) -- ({4},{-3-\error});

\draw ({-0.5},-5) node{$S_{j}$}; 
\draw ({-0.5},5) node{$S_{j}$};

\end{scope}


\begin{scope}[xshift=17cm]

\draw[thick] (-4,-3) -- (-4,3) -- (5,3) -- (5,-3) -- cycle;

\foreach \x in {-4,...,4} \foreach \y in {-3,3}
\draw[fill=black] ({0.5+\x},\y) circle [radius = 0.15];    

\foreach \x in {-4,5} \foreach \y in {-3,...,2} \draw[fill=black] (\x,{0.5+\y}) circle [radius = 0.15];       

\draw (-2.5,0) node{$A$}; 
\draw (2,0) node{$B$}; 


\draw[dashed] (-1, -3) -- (-1,3);  

\draw[|-|] ({-2},{3+\error}) -- ({0},{3+\error}); 

\draw[|-|] ({-2},{-3-\error}) -- (0,{-3-\error});

\draw ({-1},-5) node{$S_{1}$}; 
\draw ({-1},5) node{$S_{1}$};

\end{scope}


\begin{scope}

\draw[thick] (-4,-3) -- (-4,3) -- (5,3) -- (5,-3) -- cycle;

\foreach \x in {-4,...,4} \foreach \y in {-3,3}
\draw[fill=black] ({0.5+\x},\y) circle [radius = 0.15];    

\foreach \x in {-4,5} \foreach \y in {-3,...,2} \draw[fill=black] (\x,{0.5+\y}) circle [radius = 0.15];      


\draw[dashed] (-1, -3) -- (-1,3);  

\draw[|-|] ({-1.2},{-3-\error}) -- ({-4-\error},{-3-\error}) -- ({-4-\error},{3+\error}) -- (-1.2,{3+\error});

\draw[|-|] (-0.7,{-3-\error}) -- ({5+\error},{-3-\error}) -- ({5+\error},{3+\error}) -- (-0.7,{3+\error});

\draw ({-5},-5) node{$\partial A \setminus \partial B$}; 
\draw ({6},-5) node{$\partial B \setminus \partial A$}; 

\draw (-2.5,0) node{$A$}; 
\draw (2,0) node{$B$}; 

\end{scope}


\end{tikzpicture}
\caption{Adjacent rectangles $AB$.}
\label{Fig:quasiPerturbation1}
\end{figure}

We are going to consider two local Hamiltonians on $\partial AB$, namely $G$ and a perturbed version $\widetilde{G}=G+U$, and study bounds and locality properties of
\begin{equation}\label{equa:expansionalExpansion} 
e^{G}e^{-\widetilde{G}}  = \mathbbm{1} + \sum_{m=1}^{\infty} (-1)^{m} \int_{0}^{1}d \beta_{1} \ldots \int_{0}^{\beta_{m-1}}d \beta_{m} \, \Gamma_{G}^{-i \beta_{1}}(U) \ldots \Gamma_{G}^{-i \beta_{m}}(U)\,,
\end{equation}
where we have used the expansional formulas described in Section \ref{sec:expansionals}.

\begin{Theo}\label{Theo:quasilocalityAroundBorder}
Let $(G_{\partial \mathcal{R}})_{\mathcal{R}}$ be a family of local boundary Hamiltonians  with exponential decay $\| \Omega\|_{\lambda} < \infty$ for some $\lambda > 8 \Omega_{0}$. Then, for every pair of adjacent rectangles $AB$ as above it holds that the local Hamiltonians
\[ G = \sum_{X \subset \partial AB}{g_{X}^{\partial AB}}\quad \mbox{and} \quad  \widetilde{G} = \sum_{X \subset \partial A \setminus B}{g_{X}^{\partial AB}} + \sum_{X \subset \partial B \setminus \partial A}{g_{X}^{\partial AB}}\]
satisfy
\begin{align*} 
\| e^{G}e^{-\widetilde{G}}\| & \leq \exp{ (e^{12 \| \Omega\|_{\lambda}}) }\,, \\[2mm]
\| e^{G}e^{-\widetilde{G}} - e^{G_{S_{k}}} e^{\widetilde{G}_{S_{k}}}\| & \leq  \, \exp{(e^{12 \| \Omega\|_{\lambda}})} \, e^{12 \| \Omega\|_{\lambda}} \, e^{(8 \Omega_{0} - \lambda)k} \,.
\end{align*}
\end{Theo}

\begin{proof}
Let us denote by $\mathcal{F}$ the family of subsets $X$ of $\partial AB$ such that $X \cap \partial A$ and $X \cap \partial B$ are both not empty. Then,  
\[ \widetilde{G} = G + U \quad \mbox{ where } \quad U = \sum_{X \in \mathcal{F}}{-g_{X}^{\partial AB}}\,. \]
Next we split the family $\mathcal{F}$ into the following subfamilies: for every $j \geq 1$ 
\[  \mathcal{F}_{j}:= \{ X \in \mathcal{F} \colon X \subset S_{j} \} \setminus \mbox{$\bigcup_{j' < j}$} {\mathcal{F}_{j'}}\,. \]
We can then decompose
\begin{equation}\label{equa:decompositionQuasiperturbation} 
U = \sum_{j \geq 1}{U_{j}} \quad \mbox{ where } \quad U_{j}:=\sum_{X \in \mathcal{F}_{j}}{-g_{X}^{\partial AB}}\,. 
\end{equation}
Since every set $X$ in $\mathcal{F}_{j}$ has diameter greater than or equal to $j$, and it has nonempty intersection with $S_{j} \setminus S_{j-1}$ (that contains at most four points) we can easily deduce 
\[ \| U_{j}\| \leq 4 \Omega_{j}\,. \]
Moreover $U_{j}$ has support in $S_{j}$, which consists of at most two connected segments and contains at most $4j$ elements. This is an admissible supporting set, and so we can use the locality estimates from Lemma \ref{Lemm:localityEstimates Boundary} to get for $|s| \leq 1$ and every $j\leq k$
\begin{align*}
\| \Gamma_{G}^{s}(U_{j}) \| , \|\Gamma_{G_{S_{k}}}^{s}(U_{j}) \| \, & \leq \, \| U_{j}\| \, e^{2 \Omega_{0} |S_{j}|} \, e^{8 \, \| \Omega\|_{\lambda}} \, \leq \, 4 \, \Omega_{j} e^{\lambda j}  e^{8 \| \Omega\|_{\lambda}} e^{(8 \Omega_{0}-\lambda) j}\\[2mm]
\| \Gamma_{G}^{s}(U_{j}) - \Gamma_{G_{S_{k}}}^{s}(U_{j}) \| \, & \leq \,  \| U_{j}\| \, e^{2 \Omega_{0} |S_{j}|} \, e^{8 \, \| \Omega\|_{\lambda}} e^{(8 \Omega_{0} - \lambda) (k-j)}\,,\\
& \leq \, 4 \, \Omega_{j} \, e^{8 \Omega_{0} j} \, e^{8 \| \Omega\|_{\lambda}} \, e^{(8 \Omega_{0} - \lambda) (k-j)}\, = \, 4 \, \Omega_{j} \, e^{\lambda j} \, e^{8\| \Omega\|_{\lambda}} \, e^{(8\Omega_{0} - \lambda)k}\,.
\end{align*} 
Next, using \eqref{equa:decompositionQuasiperturbation} we deduce for every $k \geq 1$ and $|s| \leq 1$
\begin{align*}
\| \Gamma_{G}^{s}(U) \| \, , \, \|\Gamma_{G_{S_{k}}}^{s}(U) \| \, & \leq \, \sum_{j \geq 1} 4 \, \Omega_{j} \, e^{\lambda j} \, e^{8\| \Omega\|_{\lambda}} \, \leq \, 4 \, \| \Omega\|_{\lambda} \,  e^{8 \| \Omega\|_{\lambda}} \, \leq \, e^{12 \| \Omega\|_{\lambda}}\,
\end{align*}
and
\begin{align*}
\| \Gamma_{G}^{s}(U) - \Gamma_{G_{S_{k}}}^{s}(U_{S_{k}}) \| & \leq \, \sum_{j=1}^{k}{\| \Gamma_{G}^{s}(U_{S_{j}}) - \Gamma_{G_{S_{k}}}^{s}(U_{S_{j}}) \|} + \sum_{j >k}\| \Gamma_{G}^{s}(U_{S_{j}}) \|\\
& \leq \, 4 \, e^{8\| \Omega\|_{\lambda}} \, e^{(8 \Omega_{0} - \lambda)k}  \,\sum_{j=1}^{k} \Omega_{j} e^{\lambda j} + 4 \, e^{8 \| \Omega\|_{\lambda}} \, \sum_{j>k} \Omega_{j} e^{\lambda j} \, e^{(8 \Omega_{0}-\lambda)j}\\
&  \leq \, 4 \, \| \Omega\|_{\lambda} \, e^{8\| \Omega\|_{\lambda}}\, e^{(8 \Omega_{0} - \lambda)k}\\[3mm]
&  \leq \, e^{12 \| \Omega\|_{\lambda}}\, e^{(8 \Omega_{0} - \lambda)k} \,.
\end{align*}
Applying these estimates to \eqref{equa:expansionalExpansion} and reasoning as in the proof of Theorem \ref{Theo:BoundExpansionals} we conclude the result.
\end{proof}

\begin{Theo}\label{Theo:quasiperturbationHomogeneity}
Let $(G_{\partial \mathcal{R}})_{\mathcal{R}}$ be a family of local boundary Hamiltonians  with exponential decay $\| \Omega\|_{\lambda} < \infty$ for some $\lambda > 8 \Omega_{0}$ and $\eta$-homogeneous for an absolutely summable sequence $\eta(\ell)$. Then, for every pair of adjacent rectangles $AB$ as above it holds that 
\[ G = \sum_{X \subset (\partial A \setminus \partial B) \setminus S_{\ell}}{g_{X}^{\partial AB}}\quad\quad \mbox{and} \quad\quad  \widetilde{G} = \sum_{X \subset (\partial A \setminus \partial B) \setminus S_{\ell}}{g_{X}^{\partial A}}\]
satisfy
 \[ \| e^{G}e^{-\widetilde{G}} - \mathbbm{1}\| \leq \exp{\big[ e^{12 \| \Omega\|_{\lambda}} \mbox{$\sum_{j > \ell}$}\, \eta(j) \big]} e^{12 \| \Omega\|_{\lambda}} \Big( \mbox{$\sum_{j > \ell}$}\, \eta(j)\Big). \]
\end{Theo}

\begin{proof}
Let us denote by $\mathcal{F}$ the family of all sets $X \subset \partial A \setminus \partial B$.  We are going to split $\mathcal{F}$ into the following subfamilies: for each $1 \leq j \leq k$ let us define inductively 
\[ \mathcal{F}_{j,k}:= \{ X \in \mathcal{F} \colon  X \subset S_{k},   X \cap S_{j} \neq \emptyset \} \setminus \mbox{$\bigcup$}\{ \mathcal{F}_{j',k'} \colon j' < j \mbox{ or } k' < k \}\,.  \]
This ensures that every set $X$ in $\mathcal{F}_{j,k}$ has support contained in $S_{j,k}:=S_{k} \setminus S_{j-1}$ and diameter greater than or equal to $k-j$, since it intersects both $S_{k} \setminus S_{k-1}$ and $S_{j} \setminus S_{j-1}$. Then, we can use the locality estimates from Lemma \ref{Lemm:localityEstimates Boundary} on
\[ U_{j,k}:= \sum_{X \in \mathcal{F}_{j,k}}{g_{X}^{\partial AB} - g_{X}^{\partial A}}\, \] 
to deduce that for $|s| \leq 1$
\begin{align*} 
\| \Gamma_{G}^{s}(U_{j,k})\| \,  \leq \, \| U_{j,k}\| \, e^{2 \Omega_{0} |S_{j,k}|} \, e^{8 \| \Omega\|_{\lambda}}\, & \leq \, \| U_{j,k}\| \, e^{4 \Omega_{0} (k-j)} \, e^{8\| \Omega\|_{\lambda}}\\[2mm]
& \leq \, 4 \, \eta(j) \, \Omega_{k-j} \, e^{4 \Omega_{0}(k-j)} \, e^{8 \| \Omega\|_{\lambda}}\,
 \end{align*}
where in the last inequality we have used the $\eta$-homogeneity condition. Therefore, using that the perturbation can be written as
\[ U := \widetilde{G} - G = \sum_{\ell < j \leq k}{U_{j,k}} \] 
 we conclude that for $|s| \leq 1$
 \begin{align*} 
 \| \Gamma^{s}_{G}(U)\| \leq \sum_{\ell < j \leq k} \| \Gamma^{s}_{G}(U)\| \, & \leq \, 4 e^{8\| \Omega\|_{\lambda}} \sum_{j > \ell} \eta(j) \, \sum_{k \geq j} \Omega_{k-j} e^{4 \Omega_{0}(k-j)}\\
 & \leq \, 4 e^{8\| \Omega\|_{\lambda}} \| \Omega\|_{\lambda} \sum_{j > \ell}{\eta(j)}\\
 & \leq \, e^{12 \| \Omega\|_{\lambda}} \sum_{j>\ell}{\eta(j)}\,.
 \end{align*}
Applying these estimates to \eqref{equa:expansionalExpansion} and reasoning as in the proof of Theorem \ref{Theo:BoundExpansionals} we conclude the result.
\end{proof}

\subsection*{Main result}

\begin{Theo}
Let $(\rho_{\partial \mathcal{R}})_{\mathcal{R}}$ be a family of boundary states whose boundary Hamiltonians $(G_{\partial \mathcal{R}})_{\mathcal{R}}$ are local with exponential decay $\| \Omega\|_{\lambda}< \infty$ for some $\lambda > 8 \Omega_{0}$, and also $\eta$-homogenous for an absolutely summable sequence  $\eta(\ell)$ with sum $\| \eta\|$. Then, $(\rho_{\partial \mathcal{R}})_{\mathcal{R}}$ is approximately factorizable with function 
\[ \varepsilon(\ell) = \exp \big[ e^{K\, (1 +  \| \Omega\|_{\lambda} + \| \eta\|)} \big] \, \big(  \, e^{(8 \Omega_{0} - \lambda) \ell} \, +  \mbox{$\sum_{k > \ell}$} \,\, \eta(k) \, \big) \]
for an absolute constant $K > 0$.
\end{Theo}

\begin{proof}
Let us fix $\ell>0$ and consider an $\ell$-admissible triplet of rectangles $ABC$ as in Figure \ref{Fig:rectangulo0}. For simplicty let us denote the corresponding boundary Hamiltonians of $ABC$, $AB$, $BC$ and $B$ respectively as
\[ \rho_{ABC} = e^{2 Q_{azc}}\,, \,\,\,  \rho_{AB} = e^{2 R_{a z \gamma}}\,, \,\,\, \rho_{BC} = e^{2 S_{\alpha z c}}\,, \,\,\, \rho_{B} = e^{2 T_{\alpha z \gamma}}\,. \]
Following the same idea from \cite{KaLuPe19} one considers the following invertible matrices
\begin{align*}
\Delta_{az} & := e^{Q_{ax}} e^{-Q_{y}} e^{Q_{axy}}  \quad\quad \Upsilon_{\alpha z}:= e^{T_{\alpha x}} e^{-T_{y}} e^{T_{\alpha xy}}\\[2mm]
\Delta_{zc} & := e^{Q_{xyc}} e^{-Q_{x}} e^{Q_{yc}} \quad\quad \Upsilon_{z \gamma} :=e^{T_{xy\gamma}} e^{-T_{x}} e^{T_{y\gamma}},
\end{align*}
and shows that they provide the desired approximate factorization. We are going to explicitly check this for $\rho_{AB}$ and $\sigma_{AB}:=\Upsilon_{z\beta}\Delta_{az}$, namely 
\[ \| \rho_{AB}^{1/2} \, \sigma_{AB}^{-1} \, \rho_{AB}^{1/2} - \mathbbm{1} \| < \varepsilon(\ell)\, \]
for a suitable function $\varepsilon(\ell)$ as above. We can explicitly write
\[ \rho_{AB}^{1/2} \, \sigma_{AB}^{-1} \, \rho_{AB}^{1/2} = \left(e^{R_{axy\gamma}} e^{-Q_{axy}} e^{Q_{y}} e^{-T_{y \gamma}}\right) \left( e^{-Q_{ax}} e^{T_{x}} e^{-T_{xy\gamma}} e^{R_{axy\gamma}} \right).  \]
Let us show that each of the two factors can be suitably approximated by the identity. We deal with the first factor, the second one is analogous. Let us decompose
\[
\begin{split} 
e^{R_{axy\gamma}} e^{-Q_{axy}} e^{Q_{y}} e^{-T_{y\gamma}} \, = \,&  \big( e^{R_{axy\gamma}} e^{-R_{axy} - R_{\gamma}} \big)  \,   \big( e^{R_{axy}} e^{-Q_{axy}} \big)\\[2mm]
&  \big(e^{R_{\gamma}} e^{-T_{\gamma}}\big) \, \big(e^{Q_{y}}e^{-T_{y}}\big) \,  \big(e^{T_{y} +T_{\gamma}} e^{-T_{y \gamma}}\big)\,.
\end{split}
\] 
By Theorems
\ref{Theo:quasilocalityAroundBorder} and \ref{Theo:quasiperturbationHomogeneity}   we know that there is an absolute constant $K > 0$ such that the five factors are uniformly bounded by $$\exp{\big(\, e^{K( 1+ \| \eta\| + \| \Omega\|_{\lambda})} \, \big)}$$ 
and moreover we can approximate
\begin{align*} 
 e^{R_{axy\gamma}} e^{-R_{axy} - R_{\gamma}} \,\, & \approx \,\, e^{R_{y\gamma}} e^{-R_{y} - R_{\gamma}}\\
 e^{R_{\gamma}} e^{-T_{\gamma}}\,\,,\,\, e^{Q_{y}}e^{-T_{y}} \,\,,\,\, e^{R_{axy}} e^{-Q_{axy}} \,\,,\,\, e^{R_{y}} e^{-T_{y}}\,\,,\,\,e^{R_{y \gamma}} e^{-T_{y \gamma}} \,\, & \approx \,\, \mathbbm{1}   
\end{align*}
with an error of order
\begin{equation}\label{equa:accumulativeError}
\exp \big[ e^{K(1 +  \| \Omega\|_{\lambda} + \| \eta\|)} \big] \, \big(  \, e^{(8 \Omega_{0} - \lambda) \ell} \, +  \mbox{$\sum_{k > \ell}$} \,\, \eta(k) \, \big)\,.
\end{equation}
Thus, applying these estimates to the above decomposition we conclude that
\begin{align*}
e^{R_{axy\gamma}} e^{-Q_{axy}} e^{Q_{y}} e^{-T_{yc}}  \, &  \approx \big( e^{R_{y\gamma}}e^{-R_{y} - R_{\gamma}} \big) \, \big( e^{R_{y}} e^{-T_{y}}\big) \big( e^{R_{\gamma}} e^{- T_{\gamma}} \big) \,  \big(  e^{T_{y} + T_{\gamma}} e^{-T_{y\gamma}} \big)\\
& = e^{R_{y \gamma}} e^{-T_{y \gamma}} \approx \mathbbm{1}\,
\end{align*}
with an accumulated error of the same type \eqref{equa:accumulativeError} for a suitable absolute constant.
\end{proof}


\section{Conclusion and final remarks}

We have presented locality estimates for complex time evolution of local observables in one-dimensional systems with interactions decaying exponentially fast, and apply them to extend previous works on the clustering property of KMS states and the spectral gap problem for parent Hamiltonians of PEPS to this latter case. 

More specifically, we have shown that if the interactions on the infinite one-dimensional lattice decay as  $\Omega_{n} = O(e^{- \lambda n})$ for some $\lambda > 0$, then for every local observable $Q$ the infinite-volume time evolution operator $s \longmapsto \Gamma^{s}(Q)$ is well-defined, analytic and quasi-local (with exponential tails) on the strip $|\operatorname{Im}(s)| < \lambda/(4 \Omega_{0})$. Let us point out that we are not aware of any 1D model exhibiting this type of threshold, although there are examples in 2D  of finite range interactions with this property  \cite{Bouch15}.

We have also shown that, under the previous conditions, the infinite-volume KMS state at inverse temperature $\beta$ has exponential decay of correlations whether \mbox{$0 < \beta < \lambda/(2 \Omega_{0})$}. This leaves open the existence of phase transitions at lower temperatures, which might be unexpected according to the folklore statement that 1D systems with short range Hamiltonians do not exhibit phase transitions. A similar constraint in terms of $\lambda$ and $\Omega_{0}$ also appears in the result for PEPS. 

In both applications, these seeming thresholds arise from combining the Araki-Dyson expansional formulas for local perturbations with the prior locality estimates. Let us remark that there is an alternative perturbation formula due to Hastings \cite{Ha07} that provides $O$ satisfying $e^{-(H+U)} = O \, e^{-H} \, O^{\dagger}$ and whose locality properties depend on the locality of $U$ via the ordinary Lieb-Robinson bounds, so they are better than the ones we handle. However,  we have not succeeded in applying this formula to the previous issues, leaving this possibility open.


\bibliographystyle{plain}

\vspace{4mm}

\end{document}